\documentclass[12pt,centertags,reqno]{amsart}

\usepackage[foot]{amsaddr}
\usepackage{latexsym}
\usepackage[english]{babel}
\usepackage[T1]{fontenc}
\usepackage[utf8]{inputenc}
\usepackage[numbers]{natbib}
\usepackage{amssymb}
\usepackage{fancyhdr}
\usepackage{url}
\usepackage{hyperref}
\usepackage{verbatim}
\usepackage{leftidx}
\usepackage{color,graphicx}

\usepackage{mathrsfs, mathtools}
\usepackage{stmaryrd}

\usepackage{marvosym}
\mathtoolsset{showonlyrefs}

\usepackage{soul}

\usepackage{enumerate} 
\usepackage{enumitem} 

\numberwithin{equation}{section} \makeatletter
\renewcommand{\subsection}{\@startsection
{subsection}{2}{0mm}{\baselineskip}{-0.25cm}
{\normalfont\normalsize\bf}} \makeatother

\newtheorem{theorem}{Theorem}[section]
\newtheorem{lemma}[theorem]{Lemma}

\newtheorem{definition}[theorem]{Definition}
\newtheorem{remark}[theorem]{Remark}
\newtheorem{proposition}[theorem]{Proposition}

\newtheorem{ass}[theorem]{Assumption}

\def \A {\mathcal A}

\def \D {\mathcal D}

\def \F {\mathcal F}

\def \L {\mathcal L}

\def \P {\mathbf P}

\def \I {{\mathbf 1}}

\def \R {\mathbb R}
\def \bF {\mathbb F}

\def \bN {\mathbb N}

\newcommand{\pd}[2]{\dfrac{\partial#1}{\partial#2}}
\newcommand{\pds}[2]{\dfrac{\partial{^{2}}#1}{\partial{#2}^{2}}}
\newcommand{\pdsm}[3]{\dfrac{\partial{^{2}}#1}{\partial{#2} \partial{#3}}}

\newcommand{\ud}{\mathrm d}
\newcommand{\ds}{\displaystyle}
\newcommand{\esp}[2][\mathbb E] {#1\left[#2\right]}

\begin{document}
	
	\author[K.~Colaneri]{Katia Colaneri}\address{Katia Colaneri, Department of Economics and Finance, University of Rome Tor Vergata, Via Columbia 2, 00133 Rome, Italy.}\email{katia.colaneri@uniroma2.it}

	\author[A.~Cretarola]{Alessandra Cretarola}\address{Alessandra Cretarola, Department of Mathematics and Computer Science, University of Perugia, Via L. Vanvitelli, 1, 06123 Perugia, Italy.}\email{alessandra.cretarola@unipg.it}

	\author[B.~Salterini]{Benedetta Salterini}\address{Benedetta Salterini, Department of Mathematics and Computer Science (DIMAI), University of Firenze, Viale Morgagni, 67/a - 50134 Firenze, Italy.}\email{benedetta.salterini@unifi.it}

\title[Investment and reinsurance under forward preferences]{Optimal investment and proportional reinsurance in a regime-switching market model under forward preferences}

	\date{}
	\maketitle

\begin{abstract} \begin{center}
In this paper we study the optimal investment and reinsurance problem of an insurance company whose investment preferences are described via a forward dynamic exponential utility in a regime-switching market model. Financial and actuarial frameworks are dependent since stock prices and insurance claims vary according to a common factor given by a continuous time finite-state Markov chain. We construct the value function and we prove that it is a forward dynamic utility. Then, we characterize the investment strategy and the optimal proportional level of  reinsurance. We also perform numerical experiments and provide sensitivity analyses with respect to some model parameters. \end{center}\end{abstract}

\noindent {\bf Keywords}: Forward dynamic utility, Optimal investment, Optimal proportional reinsurance, Stochastic factor-model, Stochastic optimization.

\noindent {\bf 2020 MSC}: 60G55, 60J60, 91B30, 93E20.

\noindent {\bf JEL}: C61, G11, G22

\section{Introduction}

In this paper we study the optimal reinsurance-investment problem in a regime-switching model for an insurance
company, whose preferences are described by {\em forward dynamic utilities}.
Under this forward-looking approach the agents can adjust their (random) preferences over time, according to the available information. One of the advantages is that it allows for a significant flexibility in incorporating changing market opportunities and agents’ attitudes in a dynamically consistent manner. This means to define the forward performance process as an adapted stochastic process parameterized by wealth and time, and constructed ``forward in time''.
The pioneers of the forward investment performance approach are the papers of \citet{musiela2004example,musiela2008optimal,MZportchoice} (see also \citet{aghalith, chen} for more recent results).
  Given the initial utility function as input parameter, the forward dynamic utility for an arbitrary upcoming investment horizon is specified by means of the solution to a suitable stochastic control problem, such that the supermartingale property holds for any admissible strategy, and the martingale property holds along the optimal strategy. The latter, allows to derive a Hamilton-Jacobi-Bellman (in short HJB) equation which permits to characterize the value process.
\\
The main difference with the classical literature of backward preferences is that it does not require to set at the initial time a
utility criterion to hold at the end of the investment horizon, say $T$. Under backward preferences, when entering the market, agents must
define their risk profile at the horizon time and consequently, the portfolio is built accordingly and they cannot adapt it to variations in market conditions or update risk preferences. {The property that a clearly specified risk profile can dynamically be targeted is an advantage of forward utility preferences.}
In the actuarial framework, optimal reinsurance and investment problems have been widely investigated for different risk models and under different criteria, especially via expected utility maximization, ruin probability minimization or mean-variance criteria. However, to the best of our knowledge, {all these contributions} only employ classical backward utilities preferences, see e.g. \citet{liu2009optimal,gu2017optimal,bec,cao2020optimal,ceci2021} and references therein. A recent application of the
forward utility approach to insurance can be found in \citet{chong2019pricing}, where {an} evaluation problem of equity-linked life insurance contracts is investigated.\\
The main novelty of this paper is to consider an investment-reinsurance problem in a regime-switching market model for an insurance company whose utility preferences are described by a forward dynamic exponential utility.
The modeling framework proposed by our paper takes into consideration possible dependence between the financial and the insurance markets via the presence of a continuous time finite state Markov chain affecting the asset price dynamics, as well {as} the claim arrival intensity. This additional stochastic factor may represent some economic or geographical conditions, natural events or pandemics, that can have a considerable impact on certain lines of business of insurance companies and also affect returns of investment portfolios. {The economic effects of catastrophic events, climate changes and pandemics, as for instance the Covid-19, on both the insurance/reinsurance business and the financial market is analyzed by a recent, but quite rich, bunch of literature, see, e.g. \citet{tasselaar,baek,just,wang}.}
In our paper we address this  modeling issue by assuming that all these exogenous events are aggregated to create different regimes.
Although considering regime-switching risk models is not unusual, see, e.g.,  \citet{liu2013optimal,jang2015optimal,chen2016constrained}, to the best of our knowledge our contribution is the first which accounts for forward dynamic preferences under dependence between the actuarial and insurance framework.
The insurance company can allocate its wealth among a money market account and a risky security,
and can buy a proportional reinsurance  to hedge its insurance risks. The risky asset price {process} follows a regime-switching constant elasticity of variance (CEV) model. As observed in \citet{ma2018}, empirical {evidence suggests} that the classical CEV model represents a good alternative to stochastic volatility {models} to describe the risky asset price, which turns out to be correlated with volatility. The dependence of the coefficients on the continuous time Markov chain adds a link with the insurance modeling framework.
{A strategy} of the insurance consists of the retention level of a proportional reinsurance and the {amounts to be invested in the financial securities}.\\
Applying the classical stochastic control approach based on the {Hamilton-Jacobi-Bellman} (HJB) equation, we obtain an analytic construction of a forward dynamic exponential utility, see Theorem \ref{familyfeu},
and characterize the optimal reinsurance-investment strategy in Proposition \ref{prop:optimal}. Our analytical findings are qualitatively discussed via a numerical analysis in a two-state Markov chain model. In particular, we underline the dependence of the optimal strategy on the Markov chain under the assumption that insurance and reinsurance premia are calculated via the intensity-adjusted variance principle introduced by \citet{bec}. We also study the difference between the backward and the forward approach by comparing optimal strategies and optimal value functions, and see how the gap varies for different values of model parameters.

The paper is organized as follows. Section \ref{sec:comb_model} introduces the mathematical framework for the financial-insurance market model. In Section \ref{sec:optim} we formulate the optimization problem and  construct {a} forward dynamic exponential utility. Section \ref{sec:results} characterizes the optimal investment and reinsurance strategy. Numerical experiments and a sensitivity analysis are provided in Section \ref{sec:numerics} and Section \ref{sec:conclusion} concludes. Finally, some technical proofs are collected in Appendix \ref{sec:tech_res}. 	

\section{A regime switching insurance-financial market model}\label{sec:comb_model}

We fix a filtered probability space $(\Omega, \F, \P; \bF)$, where $\bF=\{\F_t,\ t \ge 0\}$ is a complete and right-continuous filtration.

We introduce a continuous time Markov chain $Y=\{Y_t,\ t \ge 0\}$ with finite state space
$\mathcal E=\{e_1, \dots, e_K\}$, where $e_j$, with $j=1,\ldots, K$, denote the standard vectors of $\R^K$. Let $Q=(q_{ij})_{i,j=1, \dots, K}$ be the $K\times K$ matrix representing the switching intensity. The entries of the matrix satisfy $q_{ij}\ge 0$ for all $i \neq j$ and $q_{ii}=-\sum_{i \neq j}q_{ij}$. We also recall that $Y$ admits the following semimartingale decomposition
\begin{align*}
Y_t=Y_0+\int_0^tQY_s \ud s +M^Y_t, \quad t \geq 0
\end{align*}
where $QY_s$ is the matrix-vector product and $M^Y=\{M^Y_t, \ t \ge 0\}$ is a martingale with respect to the natural filtration of $Y$. Due to the finite state nature of the Markov chain $Y$ we also get that, for any function $f:\mathcal E\to \R$, $f(Y_t)=\sum_{j=1}^Kf_j\I_{\{Y_t=e_j\}}$, where $f_j=f(e_j)$, for all $j=1, \dots, K$.
The process $Y$ is interpreted as a common stochastic factor that affects the loss process and the {risky asset} price  as shown below, and hence introduces a certain mutual dependence between the actuarial and the pure financial framework. Such behavior of the combined financial-insurance market is, nowadays, well known. Indeed, economic and geographical conditions, natural events or pandemics have a huge impact on certain lines of business of insurance companies, but they also affect returns of portfolios. In this paper we address this  modeling issue by assuming that all these exogenous events are aggregated to create different regimes.

{To describe the losses of the insurance company we introduce a Poisson process} $N=\{N_t,\ t \ge 0\}$, where $N_t$ counts the number of claims in $[0,t]$, with stochastic intensity given by the process $\{\lambda(t,Y_{t-}), \ t \ge 0\}$,
where the function  $\lambda: [0,+ \infty) \times \mathcal E \to (0,+\infty)$, is such that
\begin{equation}\label{lambdaint}
 \int_{0}^{+\infty} \lambda(t,e_j)\ud t  < \infty,
\end{equation}
for every $j=1, \dots, K$, and $\lambda(\cdot, e_j)$ is Borel-measurable. Notice that condition \eqref{lambdaint} implies in particular that
$$\mathbb{E} \bigg[ \int_{0}^{+\infty} \lambda(t,Y_{t})\ud t \bigg]\leq \max_{j=1, \dots, K}\int_{0}^{+\infty} \lambda(t,e_j)\ud t  < \infty.$$
Moreover, we observe that the intensity process $\{\lambda(t,Y_{t-}), \ t \ge 0\}$ is $\bF$-predictable.
\begin{remark}
Condition \eqref{lambdaint} implies that $N$ is non-explosive. Furthermore,
	the compensated process $\widetilde{N}=\{ \widetilde{N}_{t}, \ t \ge 0 \}$, given by
\begin{equation}
	\label{Nmg}
	\widetilde{N}_{t}= N_{t} - \int_0^t \lambda(s, Y_{s-}) \ud s, \ \ t \ge 0,
	\end{equation}
	is an $(\bF,\P)$-martingale (see \citet[Chapter II]{BREMAUD}).
\end{remark}
Let $\{T_n\}_{n \in \bN}$ be the sequence of jump times of $N$, or equivalently the {\em claims arrival times} and let $\{Z_{n}\}_{n \in \bN}$ be a sequence of independent and identically distributed $\mathcal{F}_0$-random variables independent of $N$ and $Y$. All random variables $\{Z_{n}\}_{n \in \bN}$ have common continuous distribution function with compact support $\mathcal Z\subset [0,+\infty)$, which is denoted by $F(z)$. For any $n \in \mathbb{N}$, $Z_n$ indicates {the} claim amount at time $T_n$.

\noindent The cumulative claims process $C=\{C_t,\ t \ge 0\}$ is given by
\begin{equation}
C_{t}=\sum_{n=1}^{N_{t}} Z_{n}, \ \ \ t \ge 0.
\end{equation}

\begin{remark}\label{nu}
We can provide an equivalent representation of the claim process $C$ in terms of its jump measure $m$, as follows. Define
\begin{align}
\label{rm}
m(\ud t,\ud z)& 
= \sum_{n \in \bN} \delta_{(T_{n},Z_{n})} (\ud t,\ud z),
\end{align}
where $\delta_{(t,z)}$ is the Dirac measure at point $(t,z)\in [0,+ \infty) \times \mathcal Z$; then, we get that for every $t \ge 0$
\begin{equation}
C_{t} = \int_{0}^{t} \int_{\mathcal Z} z m(\ud s,\ud z).
\end{equation}
The measure $m$ is a random counting measure with dual predictable projection $\nu$  given by
	\begin{equation}\label{nudef}
	\nu(\ud t, \ud z) = F(\ud z)\lambda(t,Y_{t-})\ud t.
	\end{equation}
Since the process $N$ is non-explosive and claim amounts have compact support, it holds that
 \begin{equation}\label{claimfiniti}
\esp{C_t}=\esp{\int_0^{t} \int_{\mathcal Z} z m (\ud s, \ud z) } = \esp{\int_0^{t} \int_{\mathcal Z} z \lambda(s,Y_{s-}) F(\ud z) \ud s} < \infty,
\end{equation}
for every $t \ge 0$. {We refer to  \citet[Chapter VIII, Section 1]{BREMAUD} for further details.}
\end{remark}

The insurance company collects premia from selling insurance contracts and buys  reinsurance contracts to share the risk that it may not be able to carry. Since claim arrival intensity is affected by the stochastic factor $Y$, following for instance \citet{bec}, we assume that both the claim premium rate and the reinsurance premium are subject to different regimes. Precisely,
the insurance gross premium is  of the form $a(t,Y_t)$, for every $t \ge 0$,  where $a:[0,+ \infty) \times \mathcal E \to [0,+\infty)$ is a continuous function in $t \geq 0$,  for all $j=1,\dots, K$.
We consider  reinsurance contracts of proportional type, with protection level $\theta=\{\theta_{t} ,\ t \ge 0\}$, i.e. at any time $t \ge 0$, $\theta_t\in [0,1]$ represents the percentage of losses which are covered by the reinsurer. The insurance company pays to the reinsurer a premium at rate $\{b(t,Y_{t},\theta_{t}), \ t \ge 0 \}$, for some function $b:[0,+ \infty) \times \mathcal E \times [0,1] \to [0,+\infty)$, which is jointly continuous with respect to $(t,\theta)$, for every $e_j\in \mathcal E$, with $j=1,\ldots,K$.

{This structure for} the insurance and the reinsurance premia includes classical premium calculation principles, as for instance, the expected value principle and the variance premium principle, as well as recently introduced calculation principles like the intensity-adjusted variance principle, proposed in \citet{bec}. The latter premium calculation rule has the advantage that optimal reinsurance strategies are chosen according to the regime.

Insurance and reinsurance premia are assumed to satisfy the conditions listed below (see, e.g. \citet{bec}), that {translate} the classical premium properties to the case where they depend on a Markov chain.
\begin{ass}\label{ipass}
The function $b{(t,e_j,\theta)}$ has continuous partial derivatives $\pd{b(t,e_j,\theta)}{\theta}$, $\pds{b(t,e_j,\theta)}{\theta}$ in $\theta \in [0,1]$ and satisfies \begin{itemize}
 \item[$(i)$] $b(t,e_j,0)=0$, for all $(t,e_j) \in [0,+ \infty) \times \mathcal E$, since the cedent does not need to pay for a null protection;
 \item[$(ii)$] $\pd{b(t,e_j,\theta)}{\theta}\ge0$, for all $(t,e_j,\theta) \in [0,+ \infty) \times \mathcal E \times [0,1]$, because the premium is increasing with respect to the protection level;
 \item[$(iii)$] $b(t,e_j,1)>a(t,e_j)$, for all $(t,e_j) \in [0,+ \infty) \times \mathcal E$, for preventing a profit without risk;
 \end{itemize}
In the sequel, $\pd{b(t,e_j,0)}{\theta}$ and $\pd{b(t,e_j,1)}{\theta}$ should be intended as right and left derivatives, respectively.
\end{ass}

It follows from the continuity of the functions $a(t, e_j)$ with respect to $t$ and of the function $b(t, e_j, \theta)$ with respect to $(t, \theta)$, for all $j=1,\dots, K$,  and the finite state nature of the Markov chain $Y$ that for every $t \geq 0$,
\begin{equation}\label{eq:iv} |a(t,Y_{t})-b(t,Y_{t},\theta_t)| \le k(t), \ \ \P{\rm -a.s.},  \end{equation}
for some continuous function $k:[0,+\infty)\to [0, +\infty)$, since $\theta_t\in [0,1]$. In particular $\int_0^tk(s) \ud s<\infty$, for all $t \ge 0$. Moreover, the following implications hold:
\begin{equation}\label{b_intcond}
\mathbb{E} \bigg[ \int_{0}^{t} b(s,Y_s,\theta_s) \ud s \bigg] \leq \max_{\theta\in [0,1], \atop j=1,\dots,K}\int_0^t b(s, e_j, \theta) \ud s < \infty,
\end{equation}
for every $t \ge 0$,  and
\begin{equation}
\label{a_intcond}
 \mathbb{E} \bigg[ \int_{0}^{t} a(s, Y_s) \ud s \bigg] \leq \max_{j=1,\dots,K} \int_{0}^{t} a(s, e_j) \ud s < \infty,
\end{equation}
for every $t \ge 0$.
In addition, we assume that
\begin{equation}\label{condconc}
	-\pds{b}{\theta}(t,e_i,\theta) < \gamma \lambda(t,e_i) \int_{{\mathcal Z} } e^{\gamma(1-\theta)z}z^2 F(\ud z) ,
\end{equation}
for every $(t,e_i,\theta) \in [0,+\infty) \times \mathcal E \times [0,1]$.
This condition guarantees the existence of a unique optimal reinsurance strategy
(see discussion after Proposition \ref{prop:optimal} below).


For any given reinsurance strategy $\theta$, the insurance company surplus (or reserve) process $R^\theta=\{R_t^\theta, \ t \ge 0\}$ is given by
\begin{equation}
 \label{surplus}
 \ud R_{t}^{\theta}=a(t,Y_{t})\ud t - b(t,Y_{t},\theta_{t}) \ud t - (1-\theta_{t-}) \ud C_{t}, \quad R_{0}^{\theta}=r_0>0.
\end{equation}
In particular, integrability conditions on insurance/reinsurance premiums imply that the surplus process is well defined. Notice that,  for every {$t\ge 0$},
\begin{align*}
|R_{t}^{\theta}|&=\left|\int_0^t(a(s,Y_s)\ud s - b(s,Y_s,\theta_s) )\ud s - \int_0^t (1-\theta_{s-}) \ud C_s\right|\\
&\leq \left|\int_0^t(a(s,Y_s)\ud s - b(s,Y_s,\theta_s) )\ud s \right|+ \left| \int_0^t (1-\theta_{s-}) \ud C_s\right|\\
&\leq \int_0^tk(s) \ud s + C_t, \quad \P-\text{a.s.},
\end{align*}
and hence $\esp{|R_{t}^{\theta}|}<\infty$, {for every $t \ge 0$, in view of \eqref{claimfiniti}.}

The insurance company is allowed to invest part of its premia in a financial market where investment possibilities are given by a riskless asset with value process $B = \{B_t,\ t \ge 0\}$ and a stock with price process $S=\{S_t, \ t \ge 0\}$. We assume zero interest rate, that is, $B_t=1$ for every $t \ge 0$, and that $S$ follows a regime-switching constant elasticity of variance (CEV) model, i.e.
\begin{equation}\label{s}
\ud S_{t}=S_t\left(\mu(Y_{t})\ud t + \sigma(Y_{t})S^\beta_{t} \ud W_{t}\right), \quad S_{0}=s >0,
\end{equation}
where $-1<\beta\leq 0$ is the coefficient of elasticity and $W=\{W_t,\ t \ge 0\}$ is a standard {$\bF$-}Brownian motion independent of $Y$ and the random measure $m(\ud t, \ud z)$.

{Specifically, we assume that the filtration $\bF=\{\F_t, t \ge 0\}$ is the completed and right-continuous filtration with
$$
\F_t=\F^W_t\vee\F^Y_t\vee\F^m_t\vee\mathcal{O},
$$
for every $t \geq 0$, where $\F^W_t=\sigma\{W_s, \ 0 \leq s \leq t\}$, $\F^Y_t=\sigma\{Y_s, \ 0 \leq s \leq t\}$, $\F^m_t=\sigma\{m([0,s)\times A), \ 0 \leq s \leq t,\ A \in \mathcal{B}(\mathcal Z)\}$ (see, e.g. \citet[Chapter VIII, Equation (1.3)]{BREMAUD}), with $\mathcal{B}(\mathcal Z)$ being the Borel $\sigma$-algebra on $\mathcal Z$, and $\mathcal O$ is the collection of $\P$-null sets.}

{The functions} $\mu: \mathcal E \to \R$ and $\sigma: \mathcal E \to [0,+\infty)$ {are measurable functions representing} the appreciation rate and the volatility of the stock, respectively. We also assume that the diffusion term is not degenerate, that is, $\sigma(e_j)> 0$ for every $j=1,\dots, K$. {Notice that functions $\mu$ and $\sigma$ may only take a finite number of values and therefore, they are bounded from above and below; in particular, it holds that $\underline{\mu}\leq \mu(Y_t)\leq \overline{\mu}$ and $0<\underline{\sigma}\leq \sigma(Y_t)\leq\overline{\sigma}$, for every $t \ge 0$, where $\underline{\mu}=\min_{j=1,\dots,K}\mu(e_j)$, $\overline{\mu}=\max_{j=1,\dots,K}\mu(e_j)$, $\underline{\sigma}=\min_{j=1,\dots,K}\sigma(e_j)$, $\overline{\sigma}=\max_{j=1,\dots,K}\sigma(e_j)$. Consequently, the ratio $\frac{\mu(Y_t)}{\sigma(Y_t)}$ is also bounded from above and below for every $t \ge 0$.}

\begin{remark}
The choice of a CEV model for the stock price process deserves a few considerations. This model was originally introduced in the paper \citet{cox1975valuation} under the assumption $\beta<0$ and extended later on to the case $\beta>0$, see e.g. \citet{emanuel1982futher}. The range of the parameter has different interpretations. For specific choices of $\beta\in \mathbb{R}$ the stock price dynamics reduces to well known processes. For instance, when $\beta=0$ and the coefficients $\mu$ and $\sigma$ are constant, we get the classical Black \& Scholes model, for $\beta=-\frac{1}{2}$ we end up with a Cox-Ingersoll-Ross process and when $\beta=-1$ the process $S$ becomes an Ornstein-{Uhlenbeck} process. It is therefore clear that, depending on the values of $\beta<0$, the process $S$ may touch zero with positive probability in finite time and even become negative, which is an unpleasant characteristic for modeling stock prices. On the other hand, if $\beta>0$ the process $S$ may explode.
Both choices for the range of $\beta$, either $\beta<0$ or $\beta>0$ have advantages and drawbacks. The setting with $-1<\beta\leq 0$ and constant coefficients $\mu$ and $\sigma$ is studied deeply for instance in \citet{delbaen2002note}, where the existence of an equivalent martingale measure is proved and considerations on absence of arbitrage are provided. For further details we also refer to, e.g.  \citet{dias2020note,heath2000martingales}.
\end{remark}

\section{{Forward exponential utility preferences}}\label{sec:optim}

We consider the problem of an insurance company with an initial wealth $x_0$, which invests its surplus in the financial market outlined in Section \ref{sec:comb_model} and buys a proportional reinsurance. For every $t \ge 0$, we denote by $\Pi_{t}$ the total amount of wealth invested in the risky asset at time $t$, and hence $X_t - \Pi_t$ is the capital invested in the riskless asset at time $t$. We assume that short-selling and borrowing from the bank account are allowed and accordingly we take $\Pi_t\in \R$ for every $t \ge 0$. {Moreover, for every $t \ge 0$, let $\theta_t\in [0,1]$ be the dynamic retention level at time $t$ corresponding to the reinsurance contract. We will consider only self-financing strategies.}
{Then, the wealth of the insurance company associated with the investment-reinsurance strategy $H=(\Pi,\theta)= \{(\Pi_{t}, \theta_{t}), \ t \ge 0\}$}
satisfies the following stochastic differential equation (in short SDE)
\begin{align}
\ud X_{t}^H &= \ud R_{t}^\theta + \Pi_{t} \frac{\ud S_{t}}{S_{t}} + (X_{t}^H-\Pi_{t})\frac{\ud B_{t}}{B_{t}} \\
& = \big{\{}  a(t,Y_{t}) - b(t,Y_{t},\theta_{t}) +\Pi_{t}\mu(Y_{t}) \big{\}}  \ud t + \Pi_{t}\sigma(Y_{t})S^\beta_t \ud W_{t} - (1-\theta_{t-})\ud C_{t},\label{wealth}
\end{align}
with $X_{0}^H=x_0 \ge 0$.
\begin{remark}
The solution of the SDE \eqref{wealth} is given by
\begin{align}
	X_t^H =& x_0 +\int_0^t  \left(a(s,Y_s) - b(s,Y_s,\theta_s)+ \Pi_s\mu(Y_s) \right)\ud s  	\\ &+ \int_0^t \Pi_s \sigma(Y_s) S_s^\beta \ud W_s - \int_0^t \int_{{\mathcal Z} }  (1-\theta_{s-}) z m(\ud s,\ud z), \quad t \ge 0.	\label{wealthsol}
\end{align}
\end{remark}

\begin{remark}
Since the insurance company can borrow a potentially unlimited amount from the bank account, the wealth process is allowed to be negative, and hence this permits us to neglect bankruptcy. 
From the mathematical point of view, dealing with a negative wealth is not a problem since we consider forward utilities of exponential type. In real life, large companies may easily have access to large amount of liquidity. Moreover, as observed also in \citet{schmidli2017risk} ``... The event of ruin almost never occurs in practice. If an insurance company
observes that their surplus is decreasing they will immediately increase their premia.
On the other hand an insurance company is built up on different portfolios. Ruin in one portfolio does not mean bankruptcy.''
The insurance company, in fact, can adjust premia dynamically. This is accounted, in our model, for instance by assuming that insurance (and reinsurance) premia are time-dependent and chosen according to the current regime.
\end{remark}

%



We assume that the preferences of the insurance company are exponential but they are specified forward in time and then, its goal is to maximize the expected forward utility, as explained below. As a first step, we provide the definition of dynamic performance process.

\begin{definition}\label{def:FDU}
Fix a normalization point $t_0\ge 0$. An $\bF$-adapted process $U(x, t_0)=\{U_t(x, t_0),\ t \ge t_0\}$ is a dynamic performance process (normalized at $t_0$) if
\begin{itemize}
\item[a)] the function $x \to U_t(x,t_0)$ is increasing and concave for all $t \ge t_0$;
\item[b)] for every self-financing strategy $H$, and for all $t, T$ such that $ t_0\leq t\leq T$ it holds that
\[
U_ {t}(X_t^H,t_0)\ge \esp{U_{T}(X_{T}^{H},t_0) | \F_{t}};
\]
\item[c)] there exists a self-financing strategy $H^*$ such that,  for all $t, T$ such that $ t_0\leq t\leq T$, it holds that
\[
U_ {t}(X_t^{H^*},t_0)= \esp{U_{T}(X_{T}^{H^*},t_0) | \F_{t}};
\]
\item[d)] at $t=t_0$,
\[
U_{t_0}(x, t_0)=u_0(x),
\]
where $u_0(x)$ is a concave and increasing function of wealth.
\end{itemize}
\end{definition}

From now on the time point $t_0$ will be our starting point and all processes and filtrations will be considered for $t \geq t_0$.

We will work under exponential preferences, that is, $u_0(x)=-e^{-\gamma x}$, with $\gamma>0$. Then, in this case Definition \ref{def:FDU} describes a forward dynamic exponential utility and can be re-formulated as follows.
\begin{definition}\label{FDU}
Let $t_0\ge0$. An $\bF$-adapted stochastic process $U=\{U_{t}(x,t_0) :\ t \ge t_0\}$ is a {\em forward dynamic exponential utility (FDU)}, normalized at $t_0$, if for all $t, T$ such that $t_0 \le t \le T$, it satisfies the stochastic optimization criterion
\begin{equation}\label{FDUsm}
		U_ {t}(x,t_0) = \left \{\begin{array}{ll}
		-e^{-\gamma x}, & t=t_0,  \\
		\sup_{{H \in}\A} \esp{U_{T}(X_{T}^{H},t_0) | \F_{t}}, &  t_0 < t \leq T,
		\end{array} \right.
\end{equation}
with $X^H$ given by \eqref{wealth}, $X^H_{t_0}=x \in \R$ and $\gamma>0$, {for a suitable class $\mathcal A$ of admissible strategies which is characterized later}.
\end{definition}

The rationale behind this definition is that at a certain time $t_0$ (for instance, $t_0=0$), the insurance company specifies its utility which is based on the available information. As {time} goes by, market conditions may change and hence the insurance company might be willing to modify its preferences accordingly. The advantages of this approach are double: {first, there is no need to specify {\em a priori} a utility to be valid at the maturity, i.e. the investor does not fix today investment preferences that will hold at a future date; second,}
defining the function $v$ as
\begin{equation}
	v(t,x,s,e_i;t_0)=\sup_{H \in \A} \mathbb{E}_{t,x,s,e_i} \Big{[} U_T(X_T^H,t_0)\Big{]},
\end{equation}
where $ \mathbb{E}_{t,x,s,e_i}[\cdot] $ denotes the conditional expectation given $X_{t}^H=x$, $S_t=s$ and $Y_{t}=e_i$, for every $(t,x,s,e_i) \in [t_0,T] \times \R \times (0,+\infty) \times \mathcal E $ and every $T \ge t_0${, it holds that}
\begin{equation}\label{sg}
{	U_t(x,t_0)=v(t,x,s,e_i;t_0).}
\end{equation}
{The latter implies that the forward dynamic exponential utility coincides with value function of the optimization problem at any time.}
An important property of the forward approach is that a forward dynamic utility might not be unique, as argued, for instance in \citet{musiela2008optimal}.

Let $t_0 \ge 0$ be the normalization point. We define the function $g: [t_0,+ \infty) \times (0,+\infty) \times \mathcal E\to \R$ as
\begin{equation}\label{g}
	g(t,s, e_i)= \frac{1}{2}\bigg( \frac{\mu(e_i)}{\sigma(e_i) s^\beta} \bigg)^2 + \gamma a(t,e_i) -\varphi(t,e_i),
\end{equation} where the function $\varphi:[t_0,+\infty) \times \mathcal E\to \R$ is given by
\begin{equation}\label{eq:phi}
	\varphi(t,e_i) = \gamma b(t,e_i,\overline{\theta}_{t}) + \lambda(t,e_i) \int_{\mathcal Z}  \Big(e^{\gamma (1-\overline{\theta}_{t-})z} -1 \Big) F(\ud z),
\end{equation}
and {$\overline{\theta}_t=\overline{\theta}(t, Y_t)$ satisfies:}
\begin{equation}\label{optr}
\overline{\theta}(t,e_i) = \left \{\begin{array}{lll}
0, \qquad \quad \ (t,e_i) \in \D_0 \\
\hat{\theta}(t,e_i), \quad (t,e_i) \in (\D_0\cup \D_1)^c \\
1, \qquad \quad \ (t,e_i) \in \D_1,
\end{array} \right.
\end{equation}
where $(\D_0\cup \D_1)^c$ indicates the complementary set of $\D_0\cup \D_1$, \begin{equation}\begin{split}
& \D_0 \equiv \bigg{\{} (t,e_i) \in [t_0,+\infty) \times \mathcal{E} \ \big{|} \ \lambda(t,e_i) \esp{Z_1 e^{\gamma Z_1}} \le \pd{b}{\theta}(t,e_i,0)   \bigg{\}} \\
& \D_1 \equiv \bigg{\{} (t,e_i) \in [t_0,+\infty) \times \mathcal{E} \ \big{|} \ \pd{b}{\theta}(t,e_i,1) \le \lambda(t,e_i)\esp{Z_1} \bigg{\}}
\end{split}
\end{equation}
{and $\hat{\theta}$ is the unique solution of the equation: }
\begin{equation}\label{solr}
 \pd{b}{\theta}(t,e_i,\theta)  = \lambda(t,e_i) \int_{{\mathcal Z} } z e^{\gamma (1-\theta)z} F(\ud z) .
\end{equation}

{The existence of a unique solution to the equation \eqref{solr} is guaranteed by the concavity assumption in equation \eqref{condconc}, see also the comment after Proposition \ref{prop:optimal}.
The process $\overline{\theta}=\{\overline{\theta}_t,\ t \ge t_0\}$ defined above takes values in $[0,1]$ and has the same structure of the optimal retention level in the standard backward utility maximization. We will see later that it also provide the optimal reinsurance strategy under forward utility preferences.}

Next, for $t\ge t_0$, we define the process $\{h(t_0,t),\ t \ge t_0\}$ as
\begin{equation}\label{h}
 h(t_0,t)= \int_{t_0}^{t} g(v,S_v,Y_v) \ud v,
\end{equation}
with $g$ given in \eqref{g}.
{
Our objective is to characterize the forward dynamic exponential utility (\textbf{Problem 1}), i.e.
$$
U_t(x,t_0)=\sup_{H\in \mathcal A}  \esp{-e^{-\gamma X_{T}^H + h(t_0,T)} \Big{|} \F_t},
$$
}
for every $t_0\le t\leq T$, and to find the optimal strategy $H \in \mathcal A$, where the set  of admissible strategies $\mathcal A$ is defined below.
\begin{definition}\label{def:admissible_strategies}
An {\em admissible} strategy is {a pair $H=(\Pi,\theta)= \{(\Pi_{t}, \theta_{t}), \ t \ge t_0\}$ of $\bF$-progressively measurable processes with values in $\R \times [0,1]$, such that, for every $T\ge t_0$ $\esp{e^{-\gamma X_T^H+h(t_0,T)}}< \infty$} and
\begin{equation} \label{int_ammiss}
\esp{\int_{t_0}^t\left(|\Pi_s| +\Pi_s^2 S_s^{2\beta}\right) \ud s} < \infty.
\end{equation}
\end{definition}
\noindent {Denote by $C_b^{1,2,2}$ the set of all bounded functions $f(t,x,s,e_j)$, with bounded first-order derivatives with respect to $t,x,s$ and bounded second-order derivatives with respect to $x,s$, for every $j =1,\ldots,K$.
Let $\L^H$ denote the Markov generator of $(X^H,S,Y)$ associated with a constant control $H=(\theta,\Pi) \in [0,1] \times \R$.}
\begin{lemma}
\label{xy}
Let $f(\cdot,\cdot,\cdot,e_i) \in C_b^{1,2,2}$, for each $ e_i \in \mathcal E$. For any constant strategy $H=(\Pi, \theta)\in \mathbb{R}\times [0,1]$, the triplet $(X^H,S,Y)$ is a Markov process 
with infinitesimal generator $\L^H$ given by
\begin{align}
&\L^H f(t,x,s,e_i)=  \pd{f}{t}(t,x,s,e_i) + \big[ a(t,e_i) - b(t,e_i,\theta) + \Pi\mu(e_i)\big]\pd{f}{x}(t,x,s,e_i) \\&\quad + \sum_{j =1}^K f(t,x,s,e_j)q_{ij}   + s \mu(e_i)\pd{f}{s}(t,x,s,e_i)+\frac{1}{2}\Pi^{2} \sigma^2(e_i) s^{2\beta} \pds{f}{x}(t,x,s,e_i)\\
&\quad +s^{2\beta+2}\sigma(e_i)\pds{f}{s}(t,x,s,e_i) + \Pi \sigma^2(e_i) s^{2\beta+1} \frac{\partial^2 f}{\partial x\partial s}(t,x,s,e_i)\\
&\quad + \lambda(t,e_i) \int_{\mathcal Z} \Big{\{}f\big(t,x-(1-\theta)z,s,e_i\big)-f(t,x,s,e_i)\Big{\}} F(\ud z). \label{mgenxy}
\end{align}
 \end{lemma}
\noindent The proof of this result is given in Appendix \ref{app:xy}.

\noindent Now, we provide the analytic construction of a forward dynamic utility in this framework.\
\begin{theorem}
	\label{familyfeu}
	Let $t_0 \ge 0$ be the forward normalization point. Then, the process $\{U_t(x,t_0),\ t \ge t_0 \}$, given for $x \in \R$ and $t \ge t_0$, by
	\begin{equation} \label{FDUe}
	U_t(x,t_0)=-e^{-\gamma x + h(t_0,t)},
	\end{equation} with the process $\{h(t_0,t),\ t \ge t_0\}$ defined in \eqref{h}, is a forward dynamic exponential utility, normalized at $t_0$.
\end{theorem}

\begin{proof}
We show that the process $\{U_t(x,t_0),\ t \ge t_0 \}$ introduced in \eqref{FDUe} satisfies Definition \ref{FDU} (equivalently, Definition \ref{def:FDU} with the initial condition $u_0(x)=-e^{-\gamma x}$). Firstly, we see that $U_t(x,t_0)$ is $\F_t$-measurable for each $t \ge t_0$ and normalized at $t_0$, indeed the condition at $t=t_0$ is satisfied (i.e. $U_{t_0}(x,t_0)=-e^{-\gamma x}$).
Next, we need to prove that for arbitrary $t, T$ such that $t_0 \le t \le T$,
\begin{equation}\label{eq:1}
-e^{-\gamma x + h(t_0,t)} =\sup_{H\in \mathcal A}  \esp{-e^{-\gamma X_{T}^H + h(t_0,T)} \Big{|} \F_t}.
\end{equation}
This means that for any self-financing strategy $H$ we get
\[-e^{-\gamma x + h(t_0,t)} \ge  \esp{-e^{-\gamma X_{T}^H + h(t_0,T)} \Big{|} \F_t},
\]
and { we will also show that there is a self financing strategy $H^*\in \mathcal A$ such that equality holds.}

We notice that {equation \eqref{eq:1}} is equivalent to say that
\begin{equation}\label{uproof}
-e^{-\gamma x} = \sup_{H \in \A} \esp{-e^{-\gamma X_{T}^H + h(t,T)} \Big{|} \F_t}.
\end{equation}
{We define the right-hand side of equation \eqref{uproof} as}
\begin{equation}\label{vfass}
	u(t,x,s,e_i) = \sup_{H \in \A} \mathbb{E}_{t,x,s,e_i} \Big[ -e^{-\gamma X_{T}^H + h(t,T)} \Big]= {\sup_{H \in \A} \mathbb{E}_{t,x,s,e_i} \Big[ -e^{-\gamma X_{T}^H + \int_t^T g(r,S_r,Y_r) \ud r} \Big],}
\end{equation}
for a function $u:[0,+\infty) \times \R \times (0,+\infty) \times \mathcal E \to (-\infty,0)$.
We proceed as follows.

\textbf{Step 1.}
{We first notice that
$$
u(t,x,s,e_i)= e^{-\int_{t_0}^t g(r,S_r,Y_r) \ud r} \sup_{H \in \mathcal A} \mathbb{E}_{t,x,s,e_i} \Big[ -e^{-\gamma X_{T}^H + \int_{t_0}^T g(r,S_r,Y_r) \ud r} \Big].
$$
 Using the martingale property of the conditional expectation, if} $u$ is sufficiently smooth (i.e. $u \in \mathcal {C}^{1,2,2}_b$), { by It\^o formula and the product rule we get that} $u$ solves the final value problem
\begin{align}
\sup_{H \in \A} \L^H u(t,x,s,e_i) + g(t,s,e_i)u(t,x,s,e_i)= 0,\label{HJBu}
\end{align}
for all $(t,x,s,e_i) \in [t_0,T) \times \R \times (0,+\infty) \times \mathcal E$, with the final condition
\begin{align}\label{HJBu1}
u(T,x,s,e_i)  =-e^{-\gamma x}, \quad & (x,s,e_i) \in \R \times (t_0,+\infty) \times \mathcal E,
\end{align}
where we recall that $\L^H$ denotes the infinitesimal generator of the Markov process $(X^H,S,Y)$ defined in \eqref{mgenxy} associated with a constant control $H$.

\textbf{Step 2.} Next, we choose $H^*=(\Pi^*, \theta^*)$ such that $\Pi^*_t=\ds \frac{\mu(Y_t)}{\gamma \sigma^2(Y_t) S_t^{2\beta}}$ and {$\theta^*_t=\overline{\theta}(t, Y_t)$ as in equation \eqref{optr}}.
We show that the function $u$ of the form
\begin{equation}\label{ansatz}
u(t,x,s,e_i)=u(x)=-e^{-\gamma x}, \quad x \in \R,
\end{equation}
is the unique solution of the problem \eqref{HJBu}--\eqref{HJBu1}. First, we note that the function $u(x)=-e^{-\gamma x}$, with $x \in \R$, solves \eqref{HJBu}--\eqref{HJBu1}.
To get uniqueness, we apply the Verification Theorem (see Theorem \ref{thver} in Appendix \ref{sec:tech_res}).  We notice that conditions ($ii$) and ($iii$) of Theorem \ref{thver} are trivially satisfied, and hence we just need to show that  
for every $t,T$ such that $t_0\le t \le T$, we have
	\begin{equation*}
	\begin{split}
	&\mathbb{E} \bigg[ \int_{t}^{T\wedge \tau_n} \Big( e^{\int_t^r g(l,S_l,Y_l) \ud l} \sigma(Y_{r})S_r^\beta\Pi_{r} \pd{u}{x}(r,X_{r}^H,S_r,Y_r) \Big)^{2} \ud r \bigg] < \infty,	
	\\  &\mathbb{E}\Big[\int_{_{\mathcal Z} } \! \int_{t}^{T\wedge \tau_n} \! e^{\int_t^r g(l,S_l,Y_l) \ud l} \Big{|} u\big{(}r,X_{r-}^H-(1-\theta_{r-})z,S_r,Y_r) \big{)} - u(r,X_{r-}^H,S_r,Y_r) \Big{|} \\
& \quad \times \lambda(r,Y_{r-}) F(\ud z) \ud r \Big] < \infty,
	\end{split}\end{equation*}
	for a suitable, non-decreasing sequence of random times $\{ \tau_n \}_{n \in \bN}$ such that $\lim_{n\to +\infty} \tau_n= +\infty$. We define the sequence $\{ \tau_n \}_{n \in \bN}$ by setting
\begin{equation}
	\tau_{n}:= \inf \Big{\{} t \ge t_0 : \ e^{\int_{t_0}^t |g(r,S_r,Y_r)| \ud r} >n \vee \ X_t^H<-n  \Big{\}}, \quad n\in \bN.
	\end{equation}
{Observe that, over the stochastic interval $\llbracket t_0, T\wedge \tau_n\rrbracket$ there is an index $\bar n \leq n$ such that
$e^{\int_{t_0}^t |g(r,S_r,Y_r)| \ud r} \leq \bar n$ and  $X_t^H<-\bar{n}$, for all $t \in \llbracket t_0, T\wedge \tau_n\rrbracket$.
We denote by $C_{\bar n}$ a constant depending on $\bar n$, which may be different from one line to another. Then,} we get that
	\begin{equation*}
	\begin{split}
	&\esp{ \int_t^{T\wedge \tau_n} \Big( e^{\int_t^r g(l,S_l,Y_l) \ud l}\Pi_{r}\sigma(Y_{r})S_r^\beta \pd{u}{x}(r,X_{r}^H,S_r,Y_r) \Big)^{2} \ud r } \\
	&  = \esp{ \int_t^{T\wedge \tau_n}  e^{2\int_t^r g(l,S_l,Y_l) \ud l}\Pi_{r}^{2}\sigma^2(Y_{r})S_r^{2\beta} \Big(\gamma e^{-\gamma X_{r}^H}  \Big)^{2} \ud r }
	 \le C_{\bar n} \esp{\int_t^{T} \Pi_{r} ^{2} S_r^{2\beta}\ud r } < \infty,
	\end{split}
	\end{equation*}
 since $\Pi$ is an admissible strategy. Moreover, we have that
\begin{equation}
	\begin{split}
	& \mathbb{E}\Big[\int_{t}^{T\wedge \tau_n} \int_{_{\mathcal Z} } e^{\int_t^r g(l,S_l,Y_l) \ud l} \Big{|} u(r,X_{r-}^H-(1-\theta_{r-})z,S_r,Y_r)  - u(r,X_{r-}^H,S_r,Y_r) \Big{|} \\
&\quad \times \lambda(r,Y_{r-}) F(\ud z) \ud r \Big] \\
	& \quad =  \esp{ \int_{t}^{T\wedge \tau_n} \int_{_{\mathcal Z}} e^{\int_t^r g(l,S_l,Y_l) \ud l} e^{-\gamma X_{r-}^H}\Big{|} e^{\gamma (1-\theta_{r-})z} -1\Big{|} \lambda(r,Y_{r-}) F(\ud z)\ud r} 
	\\  & \quad
	\le C_{\bar n} \esp{\int_{t}^{T} \int_{_{\mathcal Z} } e^{\gamma z} \lambda(r,Y_{r-}) F(\ud z) \ud r} \\ & \quad \le C_{\bar n}   \esp{\int_{t}^{T} \lambda(r,Y_{r-}) \ud r} \int_{_{\mathcal Z} } e^{\gamma z} F(\ud z) < \infty,
	\end{split}\end{equation}
	 since $\mathcal{Z} \subset [0,+\infty)$ is compact and the integrability condition \eqref{lambdaint} holds.
Therefore, thanks to Theorem \ref{thver}, the function $u(x)=-e^{-\gamma x}$ is the unique solution of the boundary problem \eqref{HJBu}--\eqref{HJBu1}.

\textbf{Step 3.} {The steps above prove that the value function $u(t,x,s,e_i)$ is given by} 
$u(x)=-e^{-\gamma x}$. {Hence, using the equality \eqref{vfass}, we get that equations \eqref{uproof} and \eqref{FDUe}  hold. Consequently, according to Definition \ref{FDU}, $U=\{U_{t}(x,t_0)= -e^{-\gamma x + h(t_0,t)},\ t \ge t_0\}$ is a forward dynamic exponential utility.}
\end{proof}

\section{Investment and reinsurance under forward dynamic exponential utilities}\label{sec:results}
In this section we characterize the optimal investment strategy and the optimal reinsurance level for the forward exponential utility in \eqref{FDUe}.

\begin{proposition}\label{prop:optimal}
Let $t_0 \ge 0$ be the forward normalization point. The optimal investment portfolio $\Pi^*_t=\Pi^*(t,S_t, Y_t)$ is given by
\begin{equation}\label{opti}
 \Pi^{*}(t,s,e_i)= \frac{\mu(e_i)}{\gamma \sigma^2(e_i) s^{2\beta}},
\end{equation}
for every $(t,s,e_i) \in [t_0,+\infty) \times (0,+\infty) \times \mathcal E$.
Assume that condition \eqref{condconc} holds for every $(t,e_i,\theta) \in [t_0,+\infty) \times \mathcal E \times [0,1]$. {Then, the process $\theta^*=\{\theta^*_t, t \ge t_0\}$, where $\theta_t^*=\overline{\theta}(t, Y_t)$ and $\overline{\theta}(t,e_j)$ is given in equation \eqref{optr}, is the optimal reinsurance level.}
\end{proposition}

{From the mathematical point of view, the condition \eqref{condconc} guarantees global concavity of the value function with respect to $\theta$, and hence that a unique maximizer exists. This condition is satisfied by the main actuarial calculation principles and it is implied for instance, by concavity of the reinsurance premium $b(t, e_j, \theta)$ with respect to the retention level $\theta$. The latter is satisfied under classical premium calculation principle and has the consequence that full reinsurance is never optimal. Considering a very general reinsurance premium described by the function $b(t,e_j, \theta)$, condition \eqref{condconc} implies that the set $\D_0$ may be non-empty, and consequently that full reinsurance may be optimal for certain time periods and certain market conditions.}

\begin{proof}
We observe that, because of the relation between the value process $U$ and the function $u(t,x,s,e_i)$ in \eqref{vfass}, we can define the functions  $\Psi^\Pi$ and $\Psi^\theta$ as
\begin{align}
&\Psi^\Pi(t,x,s,e_i, \Pi) = \Pi \mu(e_i) \pd{u}{x}(t,x,s,e_i) + \frac{1}{2} \Pi^{2}\sigma^2(e_i)s^{2\beta} \pds{u}{x}(t,x,s,e_i) \\
&\quad + \Pi \sigma^2(e_i) s^{2\beta+1} \pdsm{u}{x}{s}(t,x,s,e_i)\label{psipi}\\
&\Psi^\theta(t,x,s,e_i,\theta) = -b(t,e_i,\theta) \pd{u}{x}(t,x,s,e_i) \\
&\quad + \lambda(t,e_i)\! \int_{\mathcal{Z}}\!\! \left(u\big(t,x-(1-\theta)z,s,e_i\big)-u(t,x,s,e_i)\right) F(\ud z).\label{psid}
\end{align}
Then, for every $T \ge t_0$, the problem \eqref{HJBu}--\eqref{HJBu1} can be written as
\begin{align}\label{HJBp}
&\pd{u}{t}(t,x,s,e_i) + a(t,e_i)\pd{u}{x}(t,x,s,e_i) + \mu(e_i)s\pd{u}{s}(t,x,s,e_i) + \frac{1}{2}\sigma^2(e_i)s^{2\beta+2}\pds{u}{s}(t,x,s,e_i)\\
& + \sum_{j=1}^{K} u(t,x,s,e_i)q_{ij} + g(t,s,e_i)u(t,x,s,e_i)\\
& +  \sup_{\Pi \in \R} \Psi^\Pi(t,x,s,e_i, \Pi) +\sup_{\theta \in [0,1]}\Psi^\theta(t,x,s,e_i, \theta) =0,
\end{align}
for all $(t,x,s,e_i) \in [t_0,T) \times \R \times (0,+\infty) \times \mathcal E$ with the final condition $u(T,x,s,e_i)=-e^{-\gamma x},$ for all $(x,s,e_i) \in \R \times (0,+\infty) \times \mathcal E$.

We start with the computation of the optimal investment strategy. Since $\Psi^\Pi(t,x,s,e_i, \Pi)$ is a polynomial function in $\Pi$, from the first and the second order conditions and the form of the function $u(t,x,s,e_i)$ in equation \eqref{ansatz}, we get \eqref{opti}.

For the optimal reinsurance strategy, we apply a classical argument (see e.g. \cite[Proposition 4.1]{bec}).
Because of the assumptions on the function $b(t,e_i,\theta)$ and the smoothness of function $u(t,x,s, e_i)$ in \eqref{ansatz} with respect to $x$, $\Psi^\theta$ is continuous in $\theta \in [0,1]$ and twice continuously differentiable in $\theta \in (0,1)$, for every $(t,x,s,e_i) \in [t_0,T] \times \R \times (0,+\infty) \times \mathcal E$, for all $T \ge t_0$, and its first and second partial derivatives are given by
\begin{equation}
\pd{\Psi^\theta}{\theta}(t,x,s,e_i, \theta)= -\gamma e^{-\gamma x } \Bigg{\{}  \pd{b}{\theta}(t,e_i,\theta) - \lambda(t,e_i) \int_{{\mathcal Z}} e^{\gamma (1-\theta)z}z  F(\ud z) \Bigg{\}} ,
\end{equation}
\begin{equation}
\pds{\Psi^\theta}{\theta}(t,x,s,e_i, \theta)= - \gamma e^{-\gamma x } \Bigg{\{}  \pds{b}{\theta}(t,e_i,\theta) + \gamma  \lambda(t,e_i) \int_{{\mathcal Z} } e^{\gamma (1-\theta)z}z^2 F(\ud z) \Bigg{\}}.
\end{equation} 
By condition \eqref{condconc}, $\Psi^\theta(t,x,s,e_i, \theta)$ is also strictly concave in $\theta \in [0,1]$, and hence it admits a unique maximizer $\theta^{*} \in [0,1]$.
Next we observe that, by concavity of $\Psi^\theta$ with respect to $\theta$, the function  $\pd{\Psi^\theta}{\theta}(t,x,s,e_i, \theta)$ is decreasing in $\theta$ and it holds that
\begin{equation}\label{concavity} \pd{\Psi^\theta}{\theta}(t,x,s,e_i,1)\leq \pd{\Psi^\theta}{\theta}(t,x,s,e_i, \theta)\leq \pd{\Psi^\theta}{\theta}(t,x,s,e_i,0),
\end{equation}
for all $\theta \in (0,1)$. Then, the following cases arise:
\begin{itemize}
\item[a.] If $\Psi^\theta$ is increasing in $\theta \in [0,1]$, then the maximizer is realized for $\theta^*=1$.
\item[b.] If $\Psi^\theta$ is decreasing in $\theta \in [0,1]$, then the maximizer is realized for $\theta^*=0$.
\item[c.]If $\pd{\Psi^\theta}{\theta}(t, x,s,e_i,\hat{\theta})=0$ for some $\hat{\theta} \in [0,1]$, then $\theta^*=\hat{\theta}$.
\end{itemize}
We observe that $\Psi^\theta$ is increasing if and only if $(t, e_i) \in \D_1$. Indeed, because of concavity  of $\Psi^\theta$ with respect to $\theta$ (see  \eqref{concavity}), we get that $\pd{\Psi^\theta}{\theta}(t,x,s,e_i, \theta)>0$ is equivalent to say that $\pd{\Psi^\theta}{\theta}(t,x,s,e_i, 1)>0$. This implies that $\Psi^\theta$ is increasing if and only if $\pd{b}{\theta}(t, e_i, 1)\leq \lambda(t, e_i) \esp{Z_1}$. Equivalently $\Psi^\theta$ is increasing if and only if $(t, e_i) \in \D_0$, and finally $\pd{\Psi^\theta}{\theta}(t, x,s,e_i,\hat{\theta})=0$ corresponds to solve the equation \eqref{solr}.

It only remains to show that the process $H^*=(\Pi^*,\theta^*)$ is an admissible strategy. 
It is clear that $\theta^*_t\in [0,1]$ for every $t \ge t_0$ and that $\theta^*$ is $\bF$-adapted and c\`{a}dl\`{a}g, hence $\bF$-progressively measurable; the investment strategy $\Pi^*$ is also $\bF$-adapted and c\`{a}dl\`{a}g (hence $\bF$-progressively measurable), and for every $T \ge t_0$ it satisfies:
\begin{align}
&\esp{\int_{t_0}^T\left(|\Pi_r^*| +(\Pi_r^*)^2 S_r^{2\beta}\right) \ud r}\\
&= \esp{\int_{t_0}^T \left(\left|\frac{\mu(Y_r)}{\gamma \sigma^2(Y_r) S_r^{2\beta}}\right| + \frac{\mu^2(Y_r)}{\gamma^2 \sigma^4(Y_r) S_r^{2\beta}} \right)\ud r}\leq c \esp{\int_{t_0}^T S_r^{-2\beta} \ud r} < \infty,
\end{align}
for some constant $c>0$. The first inequality here is implied by the boundedness of $\mu(Y_t)$ and $\sigma(Y_t)$.
To show that $\esp{e^{-\gamma X^{H^*}_T+h(t_0,T)}}<\infty$, we observe that
in view of \eqref{wealthsol}, and recalling that $\theta^*_t=\overline{\theta}_t$, for every $t$, we have
\begin{align}
&	-\gamma X_T^{H^*}+h(t_0,T)= -\gamma x_{t_0} - \frac{1}{2}\int_{t_0}^T  \frac{\mu^2(Y_t)}{\sigma^2(Y_t) S_t^{2\beta}}\ud t - \int_{t_0}^T \frac{\mu(Y_t)}{\sigma(Y_t) S_t^{\beta}} \ud W_t\\& + \gamma \int_{t_0}^T \int_{{\mathcal Z} }  (1-\overline{\theta}_{t-}) z m(\ud t,\ud z)- \int_{t_0}^T\lambda(t, Y_t)\int_{{\mathcal Z} } \left(e^{\gamma (1-\overline\theta_{t-}) z}-1\right)F(\ud z),
\end{align}
where $X_{t_0}^H=x_{t_0}\in \R$. Then,
\begin{align}
&\esp{e^{-\gamma X^{H^*}_T-h(t_0,T)}}=e^{-\gamma x_{t_0}} \mathbb E \left[ e^{- \frac{1}{2}\int_{t_0}^T  \frac{\mu^2(Y_t)}{\sigma^2(Y_t) S_t^{2\beta}}\ud t - \int_{t_0}^T \frac{\mu(Y_t)}{\sigma(Y_t) S_t^{\beta}} \ud W_t}\right.\\& \left.\times e^{\gamma \int_{t_0}^T \int_{{\mathcal Z} }  (1-\overline{\theta}_{t-}) z m(\ud t,\ud z)}e^{- \int_{t_0}^T\lambda(t, Y_t)\int_{{\mathcal Z} } \left(e^{\gamma (1-\overline\theta_{t-}) z}-1\right)F(\ud z)}\right].\label{eq:equality1}
\end{align}
For $T>t_0$, we define the process $L=\{L_t, \ t \in [t_0,T]\}$ as
\[
L_t= e^{- \frac{1}{2}\int_{t_0}^t  \frac{\mu^2(Y_t)}{\sigma^2(Y_t) S_t^{2\beta}}\ud t - \int_{t_0}^t \frac{\mu(Y_t)}{\sigma(Y_t) S_t^{\beta}} \ud W_t};
\]
then, $L$ is an (integrable) $(\bF,\P)$-martingale. Precisely, $L$ is an exponential martingale with expected value equal to $1$ (See Lemma \ref{lemma:L} in Appendix \ref{sec:tech_res}) and defines an equivalent change of probability measure, i.e.  $\displaystyle L_T=\frac{\ud \widetilde \P}{\ud \P}\Big{|}_{\F_T}$. Moreover, the change of measure from $\P$ to $\widetilde \P$ does not modify the law of the Markov chain $Y$ and the compensator of the claim process $C$, since it only affects the Brownian motion $W$. This means that $Y$ and $C$ have the same law under $\P$ and under $\widetilde \P$.
The equation \eqref{eq:equality1} becomes:
\begin{align}
&e^{-\gamma x_{t_0}} \mathbb E\left[ L_T e^{\gamma \int_{t_0}^T \int_{{\mathcal Z} }  (1-\overline \theta_{t-}) z m(\ud t,\ud z)}e^{- \int_{t_0}^T\lambda(t, Y_t)\int_{{\mathcal Z} } \left(e^{\gamma (1-\overline\theta_{t-}) z}-1\right)F(\ud z)\ud t}\right]\\
&=e^{-\gamma x_{t_0}} \widetilde{\mathbb{E}}\left[  e^{\gamma \int_{t_0}^T \int_{{\mathcal Z} }  (1-\overline \theta_{t-}) z m(\ud t,\ud z)}e^{- \int_{t_0}^T\lambda(t, Y_t)\int_{{\mathcal Z} } \left(e^{\gamma (1-\overline\theta_{t-}) z}-1\right)F(\ud z)\ud t}\right]\\
&=e^{-\gamma x_{t_0}}{\mathbb{E}}\left[  e^{\gamma \int_{t_0}^T \int_{{\mathcal Z} }  (1-\overline \theta_{t-}) z m(\ud t,\ud z)}e^{- \int_{t_0}^T\lambda(t, Y_t)\int_{{\mathcal Z} } \left(e^{\gamma (1-\overline\theta_{t-}) z}-1\right)F(\ud z)\ud t}\right],
\end{align}
where $\widetilde{\mathbb{E}}[\cdot]$ denotes the expected value under the probability measure $\widetilde \P$, and in the last equality we have used that $Y$ and $C$ have the same law under $\P$ and under $\widetilde \P$.
In particular,
\[
e^{- \int_{t_0}^T\lambda(t, Y_t)\int_{{\mathcal Z} } \left(e^{\gamma (1-\overline\theta_{t-}) z}-1\right)F(\ud z)\ud t}\!\leq e^{\int_{t_0}^T\lambda(t, Y_t)\ud t}\!\leq e^{ \max_{j=1, \dots, K} \int_{t_0}^T\lambda(t, e_j)\ud t}\!:=c_T<\infty, \ \P-\text{a.s.}.
\]
Finally,
\begin{align}
&{\mathbb{E}}\left[  e^{\gamma \int_{t_0}^T \int_{{\mathcal Z} }  (1-\overline{\theta}_{t-}) z m(\ud t,\ud z)}e^{- \int_{t_0}^T\lambda(t, Y_t)\int_{{\mathcal Z} } \left(e^{\gamma (1-\overline\theta_{t-}) z}-1\right)F(\ud z)\ud t} \right]\\
&\leq c_T
\mathbb E\left[  e^{\gamma \int_{t_0}^T \int_{{\mathcal Z} }   z m(\ud t,\ud z)}\right]=c_T \esp{e^{\gamma \sum_{i=1}^{N_t}Z_i}}\\
&=c_T \sum_{n \ge 0}\esp{e^{\gamma \sum_{i=1}^{N_t}Z_i}\Big{|}N_T=n}\P(N_T=n)=c_T \sum_{n \ge 0}\esp{e^{\gamma Z_1}}^n\P(N_T=n)<\infty,
\end{align}
which implies the assertion.

\end{proof}

\subsection{Independent markets and comparison with the classical backward utility approach}\label{sec:indep}

{Now, we examine the case of independent markets. This example, although simpler than the general case considered above, already contains several key characteristics that allow us to discuss some important differences between the forward and the standard backward performance criteria.}

We consider an insurance framework as in Section \ref{sec:comb_model}. The financial market, instead consists of a riskless asset with price process $B_t=1$ for all $t \ge 0$ and a risky asset with price process $S$ whose drift and volatility are not affected by the factor $Y$, and hence its  dynamics follows
\[
\ud S_t=S_t(\mu\ud t + \sigma S^{\beta}_t \ud W_t), \quad S_0=s > 0,
\]
with $\mu\in \R$ and $\sigma>0$. The results can be  easily extended to the case where drift and volatility are functions of time only.

The wealth associated to a strategy $H=(\Pi, \theta)\in \mathcal A$ is given by {$X^H=\{X_t^H,\ t \ge t_0\}$}  such that
\begin{align}
	\ud X_{t}^H &=  \big{\{}  a(t,Y_{t}) - b(t,Y_{t},\theta_{t}) +\Pi_{t}\mu) \big{\}}  \ud t + \Pi_{t}\sigma S^\beta_t \ud W_{t} - (1-\theta_{t-})\ud C_{t},\label{wealth_excomp}
\end{align}
with $X_{t_0}^H=x_{t_0} \ge 0$ being the wealth {at time $t_0$}.

\noindent We can derive the optimal investment and reinsurance strategy  {$H^*=(\Pi^*,\theta^*)$} under the forward dynamic exponential utility, which is given by $\Pi^{*}_t=\Pi^*(S_t)$ where
\begin{align}\label{eq:investment_indep}
	 \Pi^{*}(s)= \frac{\mu}{\gamma \sigma^2 s^{2\beta}},
\end{align}
and $\theta^*$ is given by equation \eqref{optr}.
The optimal value satisfies $U_t(x,t_0)=-e^{-\gamma x + h(t_0,t)},$ for all $t \ge t_0$ and $x \in \R$, where now the process $h(t_0,t)$ is given by
\[
h(t_0,t)=\int_{t_0}^t \left(\frac{1}{2}\frac{\mu^2}{\sigma^2 S_r^{2\beta}}  + \gamma a(r,Y_r) -\varphi(r,Y_r)\right) \ud r,
\]
and we recall that the function {$\varphi(t,e_i)$} is given in \eqref{eq:phi}.


Next, we would like to compare the {optimal} strategies and the value processes arising from the forward and the standard backward utility preferences. To this aim,  we derive the optimal investment and reinsurance strategy for the classical {backward} exponential utility optimization setting. We fix a time horizon $T >t_0$ which coincides with the end of the investment period, and consider the optimization problem (\textbf{Problem 2})
\begin{align}\label{pb:backward_excomp}
	\sup_{H \in \mathcal A}\mathbb{E}\left[-e^{-\gamma X^H_T}\right].
\end{align}

\begin{proposition}\label{prop:backward}
	The optimal investment and reinsurance strategy $H^{B,*}=(\Pi^{B,*}, \theta^{B,*})$ is given by
	\begin{align}\label{eq:strategia_backward_excomp}
		\Pi^{B,*}(t,s)=\frac{\mu}{\gamma \sigma^2 s^{2\beta}}-\frac{2\beta J_1(t)}{\gamma \sigma s^{2\beta}}
	\end{align}
	and $\theta^{B,*}=\overline\theta$, with $\overline \theta$ provided in equation \eqref{optr}. The optimal value function satisfies
	$$V(t,x,s,e_i)=-e^{-\gamma x+h^{B}(t,s,e_i)},
	$$
	where $h^{B}(t,s,e_i)=J_1(t)s^{-2\beta}+J_2(t, e_i)$, for every $(t,x,s,e_i)\in [t_0,T]\times \R\times (0,+\infty)\times \mathcal E$.\\
	The function $J_1(t)$ is given by $J_1(t)=\ds -\frac{\mu^2}{2\sigma^2}(T-t)$, for every $t \in [t_0,T]$ and the function $J_2(t, e_i)$ solves the following system of ODEs
	\begin{align}
		\frac{\ud J_2}{\ud t}(t, e_i)=&\gamma[a(t, e_i)-b(t, e_i, \overline \theta)] - \sum_{j=1}^Ke^{J_2(t,e_j)-J_2(t,e_i)}q_{ij}+\frac{\mu^2}{2}\beta(2\beta+1)(T-t)\\
		&-\lambda(t, e_i) \int_{\mathcal Z} \left(e^{\gamma(1-\overline \theta)z}-1\right)F(\ud z), \ t \in [t_0,T)\label{eq:J2}
	\end{align}
	with the final condition $J_2(T, e_i)=0$, for all $i=1,\dots, K$.
\end{proposition}

The proof of this result is given in Appendix \ref{app:thm_backward}. Notice that applying the transformation $\widetilde J(t, e_i)=e^{J_2(t, e_i)}$,  equation \eqref{eq:J2} can be reduced to a linear ODE.

A few considerations can be made. First of all, the standard backward approach requires that, a utility function that must hold at some future time $T$, is specified today. Instead, in the forward approach the utility is set to hold at the initial time and changes as market conditions evolve. Equivalently, the forward approach moves in the same direction of time, and therefore it may capture information about the market in a dynamic
and consistent way.

Next, we see that the forward and the backward problems share the same optimal reinsurance strategy which depends on the Markov chain $Y$, and hence it is affected by different regimes. The optimal investment strategies in the two approaches, instead, are different: in the forward case the optimal strategy consists of the myopic component only, whereas in the backward case there is an additional component which reflects the fact that the instantaneous variance of the percentage asset price change is not constant.

{As argued for instance in \citet{MZinvestment}, the value processes under backward and forward utility preferences do not coincide in general (one may always consider the special case of a  utility function that does not change over time). This is due to the fact that they are generated in  completely different ways. Indeed, in the backward case the value function process accounts for market incompleteness by estimating the future changes (this is encompassed in the function $h^B$), whereas in the forward case the value function is adjusted dynamically as time goes by, according to the arrival of new information.}

\section{Numerical experiments}\label{sec:numerics}

In this section, we employ a numerical approach to further investigate our results
and to discuss the qualitative characteristics of optimal investment and reinsurance strategies implied by our model.

We consider a two-state Markov chain $Y$, that is, $\mathcal E =\{e_1, e_2\}$; without loss of generality we may assume that $e_1$ represents a more favorable state of the combined insurance-financial market and $e_2$ is a less favorable state. In the following, we refer to $e_1$ (respectively, $e_2$) as the {\em good} (respectively, {\em bad}) state. 
The infinitesimal generator matrix $Q$ has entries $\{q_{ij}\}_{i,j\in \{1,2\}}$  such that $q_{12}>q_{21}$: this choice suggests that it is more likely for the market to switch from the good
 state to the bad state than the opposite.

For the sake of simplicity, we assume that claims arrival intensity $\lambda$ is of exponential type, i.e.
\begin{equation}\label{eq:lambda}\lambda(t,e_j)=\lambda_0 e^{k_1t+k_2(e_j)},
\end{equation}
where  $\lambda_0=e^{Y_0}$, $k_1>0$, for every $t \in [0,+\infty)$, and the function $k_2(e_j)=j\cdot k_2$ for all $j=1,2$ and some $k_2>0$. Moreover, claim size distribution is assumed to be truncated exponential. 
We assume that insurance and financial operations take place in one year, starting from today (i.e. $t_0=0$); this means that we analyze our theoretical results in a time interval $[0,T]$, with $T=1$.
Insurance and reinsurance premia are computed according to the intensity-adjusted variance principle (see \cite{bec}), and hence, they are specifically given by
\begin{align}\label{rp}
b(t,e_j,\theta)&= \lambda(t,e_j)\esp{Z_1}\theta + 2 \delta_R  \lambda(t,e_j)\esp{Z_1^2}\left(1+T\lambda(t,e_j)\right)\theta^2,\\
a(t, e_j)&=  \lambda(t,e_j)\esp{Z_1} + 2 \delta_I \lambda(t,e_j)\esp{Z_1^2}\left(1+T\lambda(t,e_j)\right), \label{ip}
\end{align}
for $j=1,2$, and $\delta_R >0$ and $\delta_I>0$ denote the reinsurance and insurance safety loading, respectively. In equations \eqref{rp} and \eqref{ip} we reported $T$ to underline the dependence on contracts maturity, which will be omitted later, plugging $T=1$.
Finally, we set the following parameter values to $q_{12}=2$, $q_{21}=1$, $k_1=0.5$, $k_2=1$ and we fix the insurance and the reinsurance safety loading to $\delta_I=0.05$ and $\delta_R=0.1$, respectively.

\subsection{Dependent markets}

We consider the general financial market which consists of a locally risk-free asset $B$ with zero interest rate and a risky asset $S$ which follows a CEV model with drift and volatility that depend on the Markov chain $Y$ as described by the SDE \eqref{s}.
According to our interpretation of states $e_1$ and $e_2$, we assume that $\mu_1 > \mu_2$ and $\sigma_1<\sigma_2$, where $\mu_j$ and $\sigma_j$ represent the expected rate of return and the volatility of the stock, respectively, in the $j$-th regime, for $j=1,2$. In fact, it is reasonable to associate to a good state for the combined market a higher rate of return and eventually smaller fluctuations, and viceversa lower rate of return and larger volatility to the bad state. This mechanism is well known in economics (see, e.g.  \citet{FRENCH} and \citet{HAMLIN}). Our framework, however, also involves the actuarial market and the interpretation of the Markov chain is not of a purely economic nature, but may also incorporate reactions to events, such as natural disasters, pandemics or even climate and environmental states, which have an impact on both insurance losses and the general trend of financial assets. Equation \eqref{eq:lambda}, for instance, shows that the common factor $Y$ affects the claim arrival intensity in a way that the average number of claims is smaller in the good state and larger in the bad state. We report the
parameters choice for the coefficients $\mu_j$ and $\sigma_j$, for $j= 1,2$ in Table \ref{tab}.

  \begin{table}[h]
	\caption{Parameter set for the rate of return and the volatility of the stock price in the two regimes.}
	\label{tab}
	\begin{tabular}{|c|cc|}
		\hline
		Regime  & $\mu$ & $\sigma$  \\
		\hline	
		$e_1$ (good)  & $0.1$ & $0.1$  \\
		$e_2$ (bad)  & $0.05$ & $0.2$ \\
		\hline	
\end{tabular}\end{table}

{To illustrate the typical sample path of an optimal strategy, we provide in Figure \ref{os} the plot of one trajectory of the optimal dynamic investment and reinsurance strategy} given by Proposition \ref{prop:optimal}. We observe that both the investment strategy and the reinsurance level exhibit jumps at switching times of the Markov chain. Moreover, if the state of the market is good, then the insurance company opts to invest more in the stock and to reinsure a greater percentage of losses. {The model specification considered in this example (i.e. the form of the intensity function, the claim size distribution and the reinsurance premium) implies that the optimal reinsurance level is given by $\theta^*_t=\hat \theta_t$ for all $t\in [0,1]$, where $\hat \theta_t=\hat \theta(t, e_j)$ is the solution of the equation}
\[
\esp{Z_1} + 4 \delta_R \theta \esp{Z_1^2} (1+\lambda(t,e_j))=\int_{\mathcal Z}ze^{\gamma (1-\theta) z}F(\ud z).
\]
{It is easily seen, using the {\em Implicit function theorem} that the derivative of $\theta^*$ with respect to time is negative for every $t$, and this explains the piecewise linearly decreasing behaviour of the optimal reinsurance level.}\footnote{{Let $\displaystyle G(t, e_j, \theta)=\esp{Z_1} + 4 \delta_R \theta\esp{Z_1^2}\left(1+T\lambda(t,e_j)\right)-\int_{\mathcal Z}ze^{\gamma (1-\theta) z}F(\ud z),$ then, for every fixed $j=1, \dots, K$ we have that
$$\frac{\ud \theta(t, e_j)}{\ud t }=-\frac{\frac{\ud G(t, e_j, \theta(t, e_j))}{\ud t}}{\frac{\ud G(t, e_j, \theta(t, e_j))}{\ud \theta}}=-\frac{4\theta(t, e_j) \delta_R \esp{Z_1^2} \lambda_0 k_1 e^{k_1 t+k_2(e_j)}}{4 \delta_R \esp{Z_1^2} \left(1+ e^{k_1 t+k_2(e_j)}\right)+\gamma\int_{\mathcal Z}z^2e^{\gamma (1-\theta(t, e_j)) z}F(\ud z)}<0.$$}}
\begin{figure}[h]
	\includegraphics[width=7.5cm]{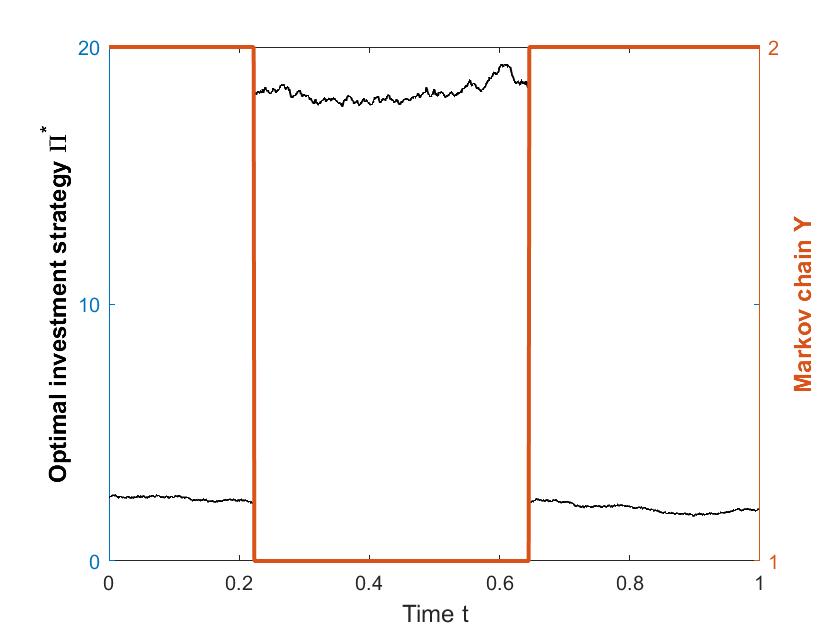}
	\includegraphics[width=7.5cm]{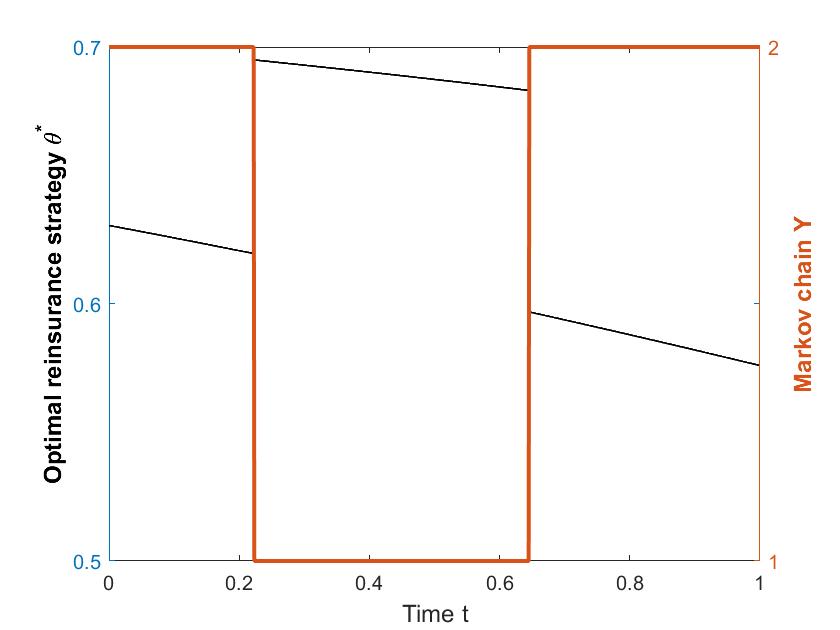}
	\caption{\label{os} The optimal investment strategy (left panel) and the optimal reinsurance strategy (right panel), as functions of time, with parameters $S_{0}=1$, $\beta=-0.5$ and $\gamma =0.5$.}
\end{figure}

In the sequel, we perform a sensitivity analysis of the optimal investment portfolio in order to study the effect of model parameters on the insurance company decision, in both economic regimes.

In Figure \ref{os(beta)}, we investigate the effect of the elasticity coefficient $\beta$ on the optimal investment strategy at a certain date $t^* \in [0,1]$. We observe that if the stock price is smaller than $1$ (left panel), the optimal investment strategy is positively correlated to the parameter of elasticity. Otherwise (right panel), the amount invested in the risky asset decreases as long as $\beta$ increases. Moreover, we observe that the investment policy is always less aggressive when the combined insurance-financial market is in the bad state (dashed lines). 

\begin{figure}[h]
	\includegraphics[width=7.5cm]{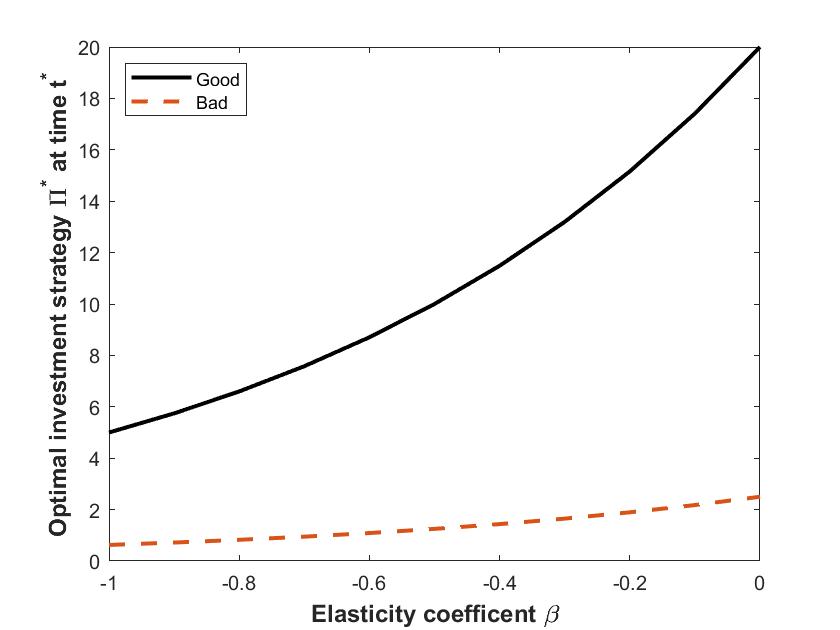}
	\includegraphics[width=7.5cm]{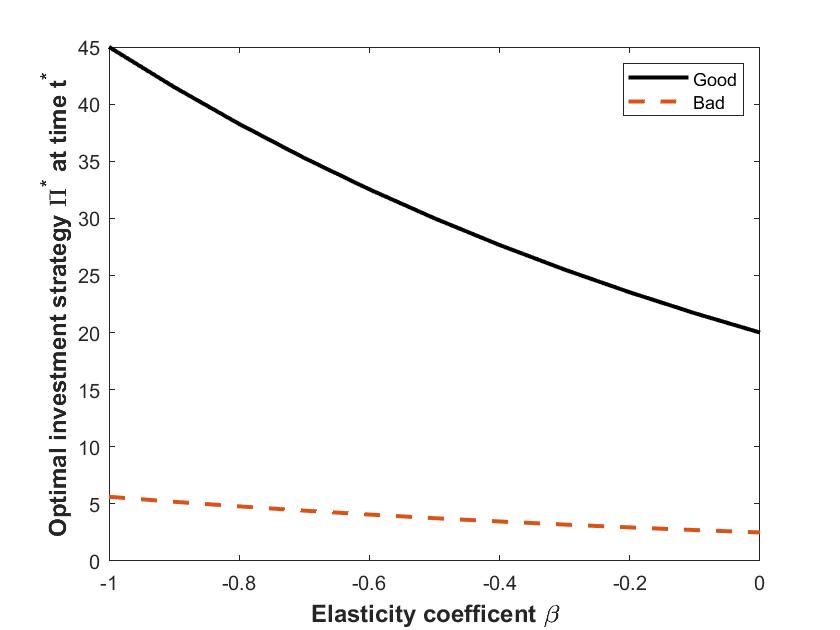}
	\caption{\label{os(beta)} Optimal investment strategy at a fixed time $t^*$ as a function of elasticity coefficient $\beta$, for different values of stock price $S_{t^*}$, when the economic regime is $e_1$ (solid line) or $e_2$ (dashed line). {Parameter} values: $\gamma=0.5$, $S_{t^*}=0.5$ (left panel) and $S_{t^*}=1.5$ (right panel).}
\end{figure}

Figure \ref{os(gamma)} refers to the sensitivity analysis with respect to the risk aversion parameter $\gamma$. As expected, there is an inverse relationship on the values of the stochastic volatility, under both regimes; in other terms, if the risk aversion increases, then the insurance company finds {it} more convenient to invest in the risk-free asset. Moreover, if the market is in the good state, the strategy is  more sensitive to variations of the risk aversion parameter.
\begin{figure}[h]
	\includegraphics[width=7.5cm]{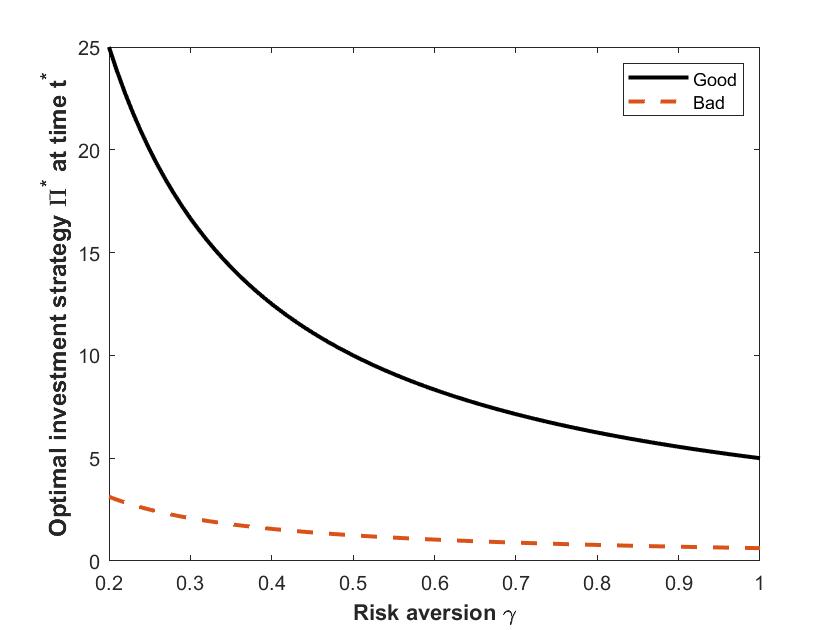}
	\includegraphics[width=7.5cm]{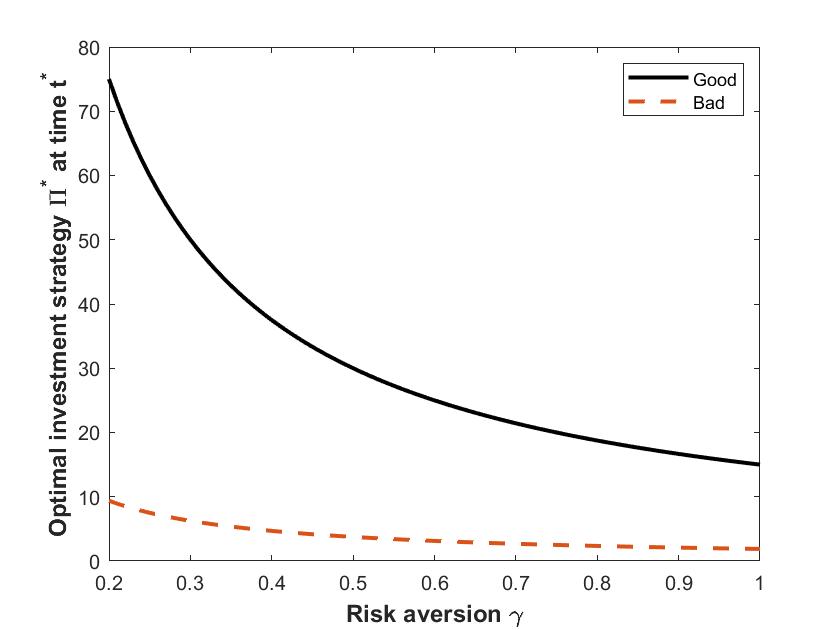}
	\caption{\label{os(gamma)} Optimal investment strategy at a fixed time $t^*$ as a function of risk aversion coefficient $\gamma$, for different values of stock price $S_{t^*}$ and constant elasticity coefficient $\beta$, when the market state is $e_1$ (solid line) or $e_2$ (dashed line). {Parameter} values: $\beta=-0.5$, $S_{t^*}=0.5$ (left panel) and $S_{t^*}=1.5$ (right panel).}
\end{figure}

\subsection{Independent markets}

{We now provide a few numerical illustrations on the case of independent markets discussed in Section \ref{sec:indep}}.

For the numerical analysis below, we take $Y$ to be a two-state Markov chain, and let  $(\bar{p},1-\bar{p})$ denote the stationary distribution of $Y$, i.e. $\bar{p}=\frac{q_{21}}{q_{12} + q_{21}}$. For consistency, we calculate the appreciation rate and the volatility of the stock price, $\mu$  and $\sigma$, as the average of the values $\mu_1, \mu_2$ and $\sigma_1, \sigma_2$, according to the stationary distribution of $Y$, that is $\mu=\bar{p}\mu_1 + (1-\bar{p}) \mu_2$ and $\sigma=\bar{p}\sigma_1 + (1-\bar{p}) \sigma_2$.

We recall that reinsurance and insurance premia are evaluated according to the intensity-adjusted variance principle and that {the} claim size distribution is exponential with expectation equal to $1$.

Figure \ref{os_comp1} and Figure \ref{os_comp} plot the difference between the strategies under forward and backward utilities, as a function of time (Figure \ref{os_comp1}), as a function of the elasticity parameter $\beta$ (Figure \ref{os_comp}, left panel) and as a function of the risk aversion coefficient $\gamma$ (Figure \ref{os_comp}, right panel). As before, we set $t_0=0$ and we consider the horizon time $T=1$. By \eqref{eq:strategia_backward_excomp}, it is clear that the optimal forward investment strategy $\Pi^*$ is more aggressive than the {backward} one $\Pi^{B,*}$ but the difference between them decreases over the time interval and it disappears at the end of trading horizon, as it is seen in Figure \ref{os_comp1}. Moreover, we notice that the higher the initial price of the risky asset is, the higher the initial gap. A similar behavior is observed in the left and in the right panel of Figure \ref{os_comp}, where we illustrate the difference in optimal initial portfolios with respect to the elasticity coefficient and the risk aversion parameter, respectively, for different initial values of the stock price.

\begin{figure}[h]
    \includegraphics[width=7.5cm]{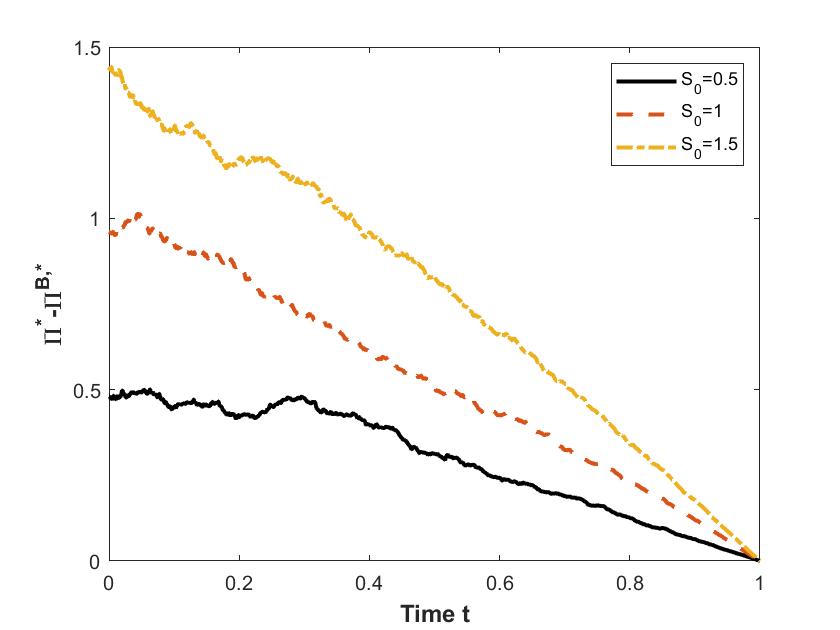}
	\caption{\label{os_comp1}One trajectory of the optimal investment as a function of time for $\beta=-0.5$ and $\gamma=0.5$.}
\end{figure}

\begin{figure}[h]
	\includegraphics[width=7.5cm]{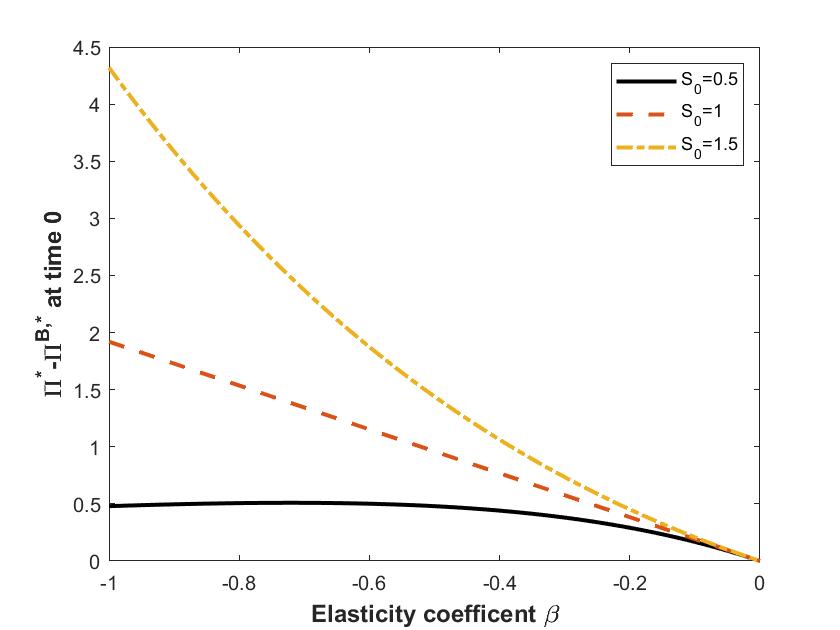}
	\includegraphics[width=7.5cm]{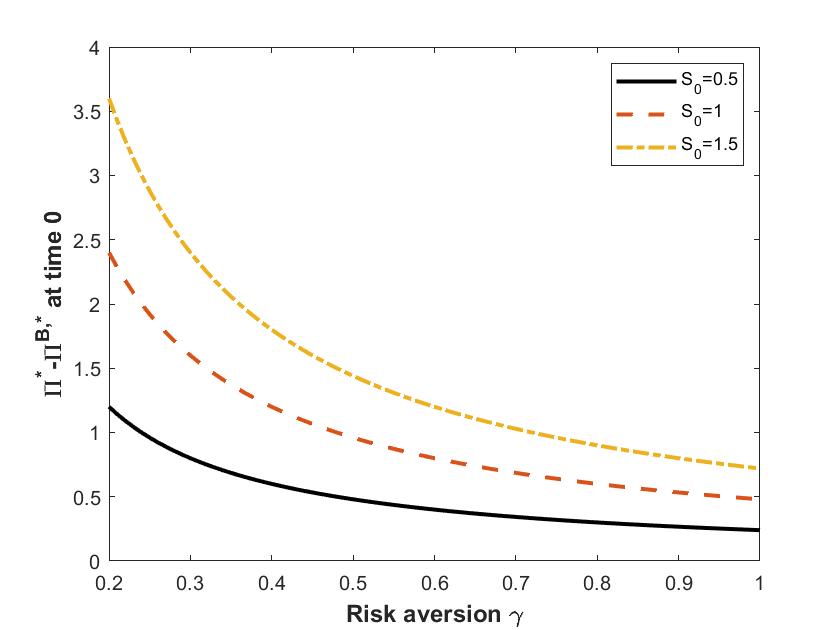}
	\caption{\label{os_comp} Left panel: Optimal investment as a function of elasticity coefficient at time $0$ for $\gamma=0.5$. Right panel: Optimal investment as a function of risk aversion parameter at time $0$ for $\beta=-0.5$.}
\end{figure}

The optimal strategies under the forward and the backward criterion lead to different value functions. In particular, at the initial time, the optimal value corresponding to the backward utility is given by  $V(0,x,s,e_i)=-e^{-\gamma x+ J_1(0)s^\beta+J_2(0, e_i)}$, whereas the optimal value in the forward utility simply satisfies $U(x, 0)=-e^{-\gamma x}$. Figure \ref{VFdiff} plots the difference between value functions at the initial time (in percentage), i.e. $\Delta(s,e_i):=\frac{V(0,x,s,e_i)-U(x,0)}{U(x,0)}$ (notice that this quantity is independent of the initial wealth), as functions  of the initial stock price, in states $e_1$ and $e_2$, assuming that the Markov chain $Y$ has only two states.
\begin{figure}[h]
	\includegraphics[width=7.5cm]{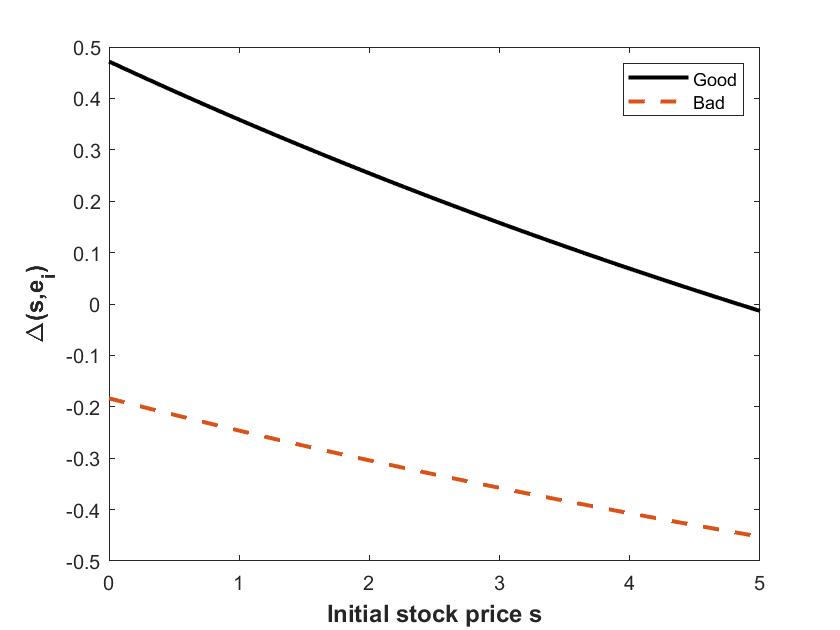}
	\caption{\label{VFdiff} The effect of stock price on the difference between the backward optimal value functions and the forward one (in percentage) at time $0$, when the market state is $e_1$ (solid line) or $e_2$ (dashed line). {Parameter} values: $\beta=-0.5$, $\gamma=0.5$}
\end{figure}
We point out that the gap between the backward and the forward values at initial time is decreasing with respect to the stock price at $t=0$, in both economic regimes. 

\section{Conclusions}\label{sec:conclusion}
{In this paper we have analyzed an optimal investment and reinsurance problem in a regime-switching market model for an insurance company with forward exponential utility preferences. 
We have proposed an interdependent insurance-financial market model where a common stochastic factor, which is modeled as a continuous time finite state Markov chain, affects the stock price and the claim arrival intensity. In this framework, we have constructed analytically a forward dynamic exponential utility and
characterized the optimal investment/reinsurance strategy.
We have also highlighted the differences between forward performance criteria and standard backward performance criteria in the case of independent markets, both analytically and numerically. 
Specifically, we have observed that the optimal reinsurance strategy is the same in the two approaches, whereas the optimal investment strategies differ, with the backward one being always smaller that the forward investment strategy. We have conducted numerical experiments, in the case of a two-state Markov chain, that confirm that the optimal forward investment strategy is more aggressive, especially for larger values of the stock price. Moreover, by comparing optimal forward and backward value functions, we have observed that the difference between them (in percentage) at initial time is decreasing with respect to the initial stock price, in both economic regimes. 
A sensitivity analysis on the general setting in case of a two-state Markov chain has 
highlighted some interesting features of the optimal forward strategies. 
We have investigated the effect of the elasticity parameter and the risk aversion coefficient on the optimal investment strategy, in both regimes. We pointed out that the relation between the amount invested in the risky asset and the parameter of elasticity also depends on the stock price, according to the CEV dynamics; whereas, the optimal investment strategy is always negatively correlated to the risk aversion coefficient.
Another interesting result is that the insurance company always opts to invest more in the risky asset when the combined insurance-financial market is in the good state (i.e. when the average number of claims is small and the stock price presents a high rate of return and small fluctuations).\\
It would be interesting to see if these features of the value function and the optimal strategies are preserved under different types of forward utilities (i.e. non zero volatility case). This will be done in a future work.}

\section*{Acknowledgments}
{The authors  are member of INDAM-Gnampa and their work has been partially supported through the Project U-UFMBAZ-2020-000791.}

\appendix

\section{Technical results}\label{sec:tech_res}





\subsection{Proof of Lemma \ref{xy}}\label{app:xy}

To prove the result, we first characterize the martingale $M^Y$ in the semimartingale decomposition of the Markov chain $Y$.  Let $\{\tau_n\}_{n \in \bN}$ be the sequence of jump times of $Y$ and denote by $m^Y$ the jump measure of $Y$, which is given by \begin{equation}
	m^Y([0,t],\{e_j\}) := \sum_{{n \ge 1}} \I_{\{ Y_{\tau_n}=e_j\}} \I_{\{\tau_n \geq t\}},
\end{equation}
with compensator
\begin{equation}
	\nu^Y([0,t],\{e_j\}) = \int_0^t \sum_{\substack{i,j=1,\\i \neq j}}^K q_{ij} \I_{\{ Y_{r-}=e_i\}} \ud r,
\end{equation}
for every $t\geq 0$. Hence, we get that
\begin{equation}
	Y_t=Y_0 + \int_0^t (e_j - Y_{r-})q_{Y_{r-}j} \ud r + \int_0^t \sum_{j=1}^{K} (e_j - Y_{r-}) (m^Y-\nu^Y) (\ud r,\{e_j\}),
\end{equation} for every $t \geq 0$\footnote{with a slight abuse of notation we identify $q_{Y_{r-}j}|_{Y_{r-}=i}=q_{ij}$.}.
Now, let $f:[0,+\infty)\times \R \times (0,+\infty)\times \mathcal E \to \R$ be a function in $C_b^{1,2,2}$ and $H=(\Pi, \theta)\in \R \times [0,1]$ constant. Then, the result follows by applying It\^{o}'s formula to $f(X^H,S,Y)$. Indeed, we get that $\{f(t,X^H_t,S_t,Y_t), t \ge 0\}$ has the semimartingale decomposition
\[
f(t,X^H_t,S_t,Y_t)=f(0,X^H_0,S_0,Y_0)+\int_0^t \mathcal L^Hf(r, X^H_r,S_r,Y_r) \ud r + M^f_t, \quad t \ge 0,
\]
where $M^f=\{M^f_t, \ t \ge 0\}$ is the $(\bF,\P)$-martingale null at $t=0$ given by
\begin{align}
&\ud M^f_t= \left(\sigma(Y_t)S_t^{\beta+1}\pd{f}{s}(t,X^H_t,S_t,Y_t) + \Pi_t \sigma(Y_t) S^{\beta}_t \pd{f}{x}(t, X^H_t,S_t,Y_t)\right)\ud W_t \\
& + \int_{\mathcal Z} \left(f\big(t,X_{t-}^H-(1-\theta)z,S_t,Y_t\big)-f(t,X_{t-}^H,S_t,Y_t)\right) \left(m(\ud t, \ud z)-\lambda(t, Y_{t-})F(\ud z)\ud t\right)\\
& + \sum_{j=1}^{K} \left(f\big(t,X_t^H,S_t,e_j\big)-f(t,X_{t}^H,S_t,Y_{t-})\right) (m^Y-\nu^Y) (\ud t , \{e_j\}).
\end{align}

\subsection{The verification theorem}\label{app:verification}

\begin{theorem}[Verification Theorem]\label{thver}
	Let $t_0\ge 0$ be the normalization point and $T \ge t_0$. Let $\overline{u}: [t_0,T] \times \R \times (0,+\infty) \times \mathcal{E} \longrightarrow (-\infty,0)$ be a smooth solution of the HJB equation \eqref{HJBu}--\eqref{HJBu1} (i.e. the function $\overline{u}(t,x,s,e_j)$ is $\mathcal C^1$ in $t$ and $\mathcal C^2$ in $(x,s)$, for all $j=1,...,K$), which satisfies
	\begin{itemize}
		\item [\rm{(i)}] $\displaystyle{\esp {\int_{t_0}^{T} e^{\int_{t_0}^r g(l,S_l,Y_l) \ud l} \Big( \Pi_{r}\sigma(Y_{r})S_v^\beta \pd{\overline{u}}{x}(r,X_{r}^H,S_r,Y_r) \Big)^{2} \ud r }< \infty}$,
		\item [\rm{(ii)}] $\displaystyle{ \esp{ \int_{t_0}^{T} e^{\int_{t_0}^r g(l,S_l,Y_l) \ud l} \Big( \sigma(Y_r)S_r^{\beta+1} \pd{\overline{u}}{s}(r,X_{r}^H,S_r,Y_r) \Big)^{2} \ud r }< \infty}$,
		\item [\rm{(iii)}] $\displaystyle{ \esp{ \int_{t_0}^{T} e^{\int_{t_0}^r g(l,S_l,Y_l) \ud l}\sum_{j=1}^{K} \Big{ \{ } \overline{u}\big{(}r,X_{r}^H,S_r,e_j \big{)} - \overline{u}(r,X_{r}^H,S_r,Y_{r-}) \Big{\}} \nu^Y(\ud r,\{e_j\}) }< \infty}$,
		\item [\rm{(iv)}] $\displaystyle \mathbb{E}\Big[\int_{t_0}^{T} \!\!\!e^{\int_{t_0}^r g(l,S_l,Y_l) \ud l} \lambda(r,Y_{r}) \\ \qquad \times \sup_{z \in \mathcal{Z}} \Big{|} \overline{u}(r,X_{r-}^H-(1-\theta_{r-})z,S_r,Y_r)-\overline{u}(r,X_{r-}^H,S_r,Y_r) \Big{|}  \ud r \Big]\! <\! \infty.$
	\end{itemize}

Then, $u(t,x,s,e_i) \le \overline{u}(t,x,s,e_i)$, for every admissible control $H \in \A$ and for every $(t,x,s,e_i) \in [t_0,T] \times \R \times (0,+\infty) \times \mathcal E$.

Moreover, if $\overline{u}(T,x,s,e_i)=u(T,x,s,e_i)$, for every $(x,s,e_i) \in \R \times (0,+\infty) \times \mathcal E$ and there exists $H^* \in \A$ such that $\L^{H^*}\overline{u}(t,x,s,e_i)+g(t,x,s,e_i)\overline{u}(t,x,s,e_i)=0$, for every $(t,x,s,e_i) \in [t_0,T[ \times \R \times (0,+\infty) \times \mathcal E$, then $u=\overline{u}$ in $[t_0,T] \times \R \times (0,+\infty) \times \mathcal E$.
\end{theorem}

\begin{proof}
	Let $H \in \A$ be an admissible control. Using equations \eqref{wealth} and \eqref{s} and applying It\^{o}'s formula to $e^{\int_{t_0}^t g(r,X_r^H,Y_r) \ud r}\overline{u}(t,X_t^H,S_t,Y_t)$, we have that
	\begin{align}
		& e^{\int_{t_0}^T g(r,S_r,Y_r) \ud r}\overline{u}(T,X_{T}^H,S_T,Y_{T}) =  e^{\int_{t_0}^t g(r,S_r,Y_r) \ud r}\overline{u}(t,x,s,e_i) \\ &  +  \int_{t}^{T} e^{\int_t^r g(l,S_l,Y_l) \ud l} \Big[ \L^H \overline{u}(v,X_v^H,S_v,Y_v) + g(r,S_r,Y_r) \overline{u}(r,X_r^H,S_r,Y_r)\Big] \ud r\\
& +\int_{t}^{T}e^{\int_t^r g(l,S_l,Y_l) \ud l} \Pi_{r}\sigma(Y_{r})S_r^\beta\pd{\overline{u}}{x}(r,X_{r}^H,S_r,Y_{r})  \ud W_{r} \\
&  + \int_{t}^{T}e^{\int_t^r g(l,S_l,Y_l) \ud l} \sigma(Y_{r})S_r^{\beta+1}\pd{\overline{u}}{s}(r,X_{r}^H,S_r,Y_{r})  \ud W_{r} \\
& + \int_{t}^{T}\int_{\mathcal{Z}} \!\! e^{\int_t^r g(l,S_l,Y_l) \ud l} \sum_{j=1}^{K}\left\{\overline{u} \big(r,X_{r}^H,S_r,e_j \big) - \overline{u}(r,X_{r}^H,S_r,Y_{r-}) \right\}  (m^Y - \nu^Y) (\ud r,\{e_j\}) \\
& + \int_{t}^{T}\int_{\mathcal{Z}}  e^{\int_t^r g(l,S_l,Y_l) \ud l} \left(\overline{u} \big(r,X_{r-}^H -(1-\theta_{r-})z,S_r,Y_r \big) - \overline{u}(r,X_{r-}^H,S_r,Y_r) \right)\\
&  \qquad \times (m(\ud r, \ud z)- \lambda(r,Y_{r-}) F(\ud z) \ud r),\label{itoV}
	\end{align}
	where $\mathcal L^H$ is introduced in \eqref{mgenxy}. Let $M=\{M_t,\ t \in [t_0,T]\}$ be the process given by
	\begin{align}
		&M_t = \int_{t_0}^{t}e^{\int_t^r g(l,S_l,Y_l) \ud l} \Pi_{r}\sigma(Y_{r})S_r^\beta\pd{\overline{u}}{x}(r,X_{r}^H,S_r,Y_{r})  \ud W_{r} \\ & + \int_{t_0}^{t}e^{\int_t^r g(l,S_l,Y_l) \ud l} \sigma(Y_{r})S_r^{\beta+1}\pd{\overline{u}}{s}(r,X_r^H,S_r,Y_{r})  \ud W_{r} \\ &  + \int_{t_0}^{t}\int_{\mathcal{Z}} \!\! e^{\int_t^r g(l,S_l,Y_l) \ud l} \sum_{j=1}^{K}\left\{\overline{u} \big(r,X_{r}^H,S_r,e_j \big) - \overline{u}(r,X_{r}^H,S_r,Y_{r-}) \right\}  (m^Y - \nu^Y) (\ud r,\{e_j\}) \\ &  + \int_{t_0}^{t}\int_{\mathcal{Z}}  e^{\int_t^r g(l,S_l,Y_l) \ud l} \left(\overline{u} \big(r,X_{r-}^H -(1-\theta_{r-})z,S_r,Y_r\big) - \overline{u}(r,X_{r-}^H,S_r,Y_r) \right)\\& \qquad  \times (m(\ud r, \ud z)- \lambda(r,Y_{r-}) F(\ud z) \ud r), 	\label{MG}
	\end{align}
and observe that integrability conditions (i), (ii), (iii), (iv) ensure that $M$ is an $(\bF,\P)$-martingale.
Now, since $\overline{u}$ solves the HJB-equation in \eqref{HJBu}--\eqref{HJBu1}, we have
\begin{align}
		& e^{\int_{t_0}^T g(r,S_r,Y_r) \ud r}\overline{u}(T,X_{T}^H,S_T,Y_{T}) \le e^{\int_{t_0}^t g(r,S_r,Y_r) \ud r}\overline{u}(t,x,s,e_i) \\
&+\int_{t}^{T}e^{\int_t^r g(l,S_l,Y_l) \ud l} \Pi_{r}\sigma(Y_{r})S_r^\beta\pd{\overline{u}}{x}(r,X_{r}^H,S_r,Y_{r})  \ud W_{r} \\ & + \int_{t}^{T}e^{\int_t^r g(l,S_l,Y_l) \ud l} \sigma(Y_{r})S_r^{\beta+1}\pd{\overline{u}}{s}(r,X_{r}^H,S_r,Y_{r})  \ud W_{r} \\ &  + \int_{t}^{T}\int_{\mathcal{Z}} \!\! e^{\int_t^r g(l,S_l,Y_l) \ud l} \sum_{j=1}^{K}\left\{\overline{u} \big(r,X_{r}^H,S_r,e_j \big) - \overline{u}(r,X_{r}^H,S_r,Y_{r-}) \right\}  (m^Y - \nu^Y) (\ud r,\{e_j\}) \\ &  + \int_{t}^{T}\int_{\mathcal{Z}}  e^{\int_t^r g(l,S_l,Y_l) \ud r} \left(\overline{u} \big(r,X_{r-}^H -(1-\theta_{r-})z,S_r,Y_r \big) - \overline{u}(r,X_{r-}^H,S_r,Y_r) \right)\\& \qquad  \times (m(\ud r, \ud z)- \lambda(r,Y_{r-}) F(\ud z) \ud r)
,\label{itoV2}
	\end{align} for every $H \in \A$.

Then, taking the conditional expectation with respect to $X_t^H=x$, $S_t=s$ and $Y_t=e_i$ on both sides of equation \eqref{itoV2} leads to
 \begin{equation}
		\mathbb{E}_{t,x,s,e_i} \Big[ e^{\int_{t_0}^T g(r,S_r,Y_r) \ud r}\overline{u}(T,X_{T}^H,S_T,Y_{T}) \Big] \le e^{\int_{t_0}^t g(r,S_r,Y_r) \ud r} \overline{u}(t,x,s,e_i).
	\end{equation} By the final condition in equation \eqref{HJBu1}, we obtain \begin{equation}
		\mathbb{E}_{t,x,s,e_i} \Big[ -e^{-\gamma X_T^H + \int_t^T g(r,S_r,Y_r) \ud r} \Big] \le \overline{u}(t,x,s,e_i),
	\end{equation} for every $H \in \A$. Hence, $u(t,x,s,e_i) \le \overline{u}(t,x,s,e_i)$, as we wanted. Finally, we observe that if $H \in \A$ is the maximizer in the HJB-equation \eqref{HJBu}, then the inequality above becomes an equality, which proves the second part of the statement.
\end{proof}

\subsection{The density $L_T$}

\begin{lemma}\label{lemma:L}
Let $T \ge 0$. Define the process $L=\{L_t, \ t \in [0,T]\}$ as
\[
L_t= e^{- \frac{1}{2}\int_{0}^t  \frac{\mu^2(Y_r)}{\sigma^2(Y_r) S_r^{2\beta}}\ud r - \int_{0}^t \frac{\mu(Y_r)}{\sigma(Y_r) S_r^{\beta}} \ud W_r};
\]
then, $L$ is an $(\bF,\P)$-martingale. Moreover, $L_T$ is the density of a probability measure $\widetilde \P$, equivalent to $\P$ on $\F_T$.
\end{lemma}
\begin{proof}
For the ease of notation we now take $t_0=0$. The proof extends that of Theorem 2.3 in \cite{delbaen2002note} to the regime-switching version of the CEV model. We summarize the main steps.
Consider the couple $(Y,S)$ where $Y$ is a finite state Markov chain with infinitesimal generator $Q$, and $S$ is a process with continuous trajectories. Consider the following equations
\begin{equation}
\ud S_t= S_t \mu(Y_t) \ud t + S^{1+\beta}_t \sigma(Y_t) \ud \overline{W}_t
\end{equation}
and 
\begin{equation}
\ud S_t= S^{1+\beta}_t \sigma(Y_t) \ud \overline{W}_t,
\end{equation}
where $\overline W$ is a Wiener measure. Next we denote by $\P$ the law of the couple $(Y,S)$, where $S$ satisfies the first equation, on the interval $[0, T ]$ and by $\widetilde \P$ the law of the couple $(Y,S)$, where $S$ satisfies the first equation, on the interval $[0, T ]$. Notice that the generator of the Markov chain $Y$ is the same under $\P$ and under $\widetilde \P$. Then, we can find a $\P$-Brownian motion  $W$ and  a $\widetilde \P$-Brownian motion $\widetilde W$, both independent of $Y$, such that
\begin{equation}
\ud S_t= S_t \mu(Y_t) \ud t + S^{1+\beta}_t \sigma(Y_t) \ud W_t
\end{equation}
and 
\begin{equation}
\ud S_t= S^{1+\beta}_t \sigma(Y_t) \ud \widetilde{W}_t.
\end{equation}
We denote by $\mathcal \bF^{Y,S}$ the filtration generated by the pair $(Y,S)$. Notice that, for instance, this coincides with the filtration generated by the processes $(Y,W)$, and moreover, because of independence of the process $(Y,S)$ with the jump measure $m(\ud t, \ud z)$, describing the jumps of the claim process, we can extend our analysis to the whole filtration $\bF$.
The laws $\P$ and $\widetilde \P$ are measures on the product space $\mathcal M \times \mathcal C$, where $\mathcal M$ is the space of piecewise continuous functions of $[0,T]$ and $\mathcal C$ is the space of continuous functions on $[0,T]$.
To show that $\P$ and $\widetilde \P$ are equivalent we define the sequence of stopping times
\[
\eta_n=\inf\left\{t>0: \ \int_0^{t} S^{-2\beta}_r \ud r\geq n\right\}.
\]
Clearly, $\eta_n \to +\infty$ (since $0\leq -\beta< 1$) and the density of $\widetilde \P$ with respect to $\P$ on $\mathcal F_{\eta_n\wedge T}$ is given by
\[
L_{\eta_n\wedge T}= e^{- \frac{1}{2}\int_{0}^{\eta_n\wedge T}  \frac{\mu^2(Y_t)}{\sigma^2(Y_t) S_t^{2\beta}}\ud t - \int_{0}^{\eta_n \wedge T} \frac{\mu(Y_t)}{\sigma(Y_t) S_t^{\beta}} \ud W_t}.
\]
Because $\int_0^T S^{-2\beta}_t \ud t<+\infty$ $\P$-a.s. and $\frac{\mu(Y_t)}{\sigma(Y_t)}$ is bounded for every $t\ge 0$, we have that $\widetilde \P$ is absolutely continuous with respect to $\P$ on $\F_T$. Conversely, we can repeat the same reasoning and use that  $\int_0^T S^{-2\beta}_t \ud t<+\infty$ $\widetilde \P$-a.s., which implies equivalence.
Hence $L$ is a strictly positive martingale with $\esp{L_T}=1$.

We also observe that, the change of measure above does not alter the law (i.e. the infinitesimal generator) of the Markov chain $Y$ nor the law (i.e. the compensator) of the jump process $C$. 
Hence $Y$ and $C$ have the same law under $\P$ and $\widetilde \P$.
\end{proof}

\subsection{Proof of Proposition \ref{prop:backward}}\label{app:thm_backward}

We notice that the optimization is taken over the set of admissible functions $\mathcal A$, even though in the backward case one would require that $\esp{e^{-\gamma X^H_T}}<\infty$ in place of $\esp{e^{-\gamma X^H_T+h(t_0,T)}}<\infty$. However, because of the assumptions on the model coefficients these two conditions are equivalent.

Suppose that the value function $V(t,x,s,e_i)$ is $\mathcal C^1$ in $t$ and $\mathcal C^2$ in $(x,s)$, for each $i=1,\ldots,K$, then it solves the equation
\begin{align}\sup_{H \in \mathcal A}\mathcal{L}^HV(t,x,s,e_i)=0, \quad (t,x,s,e_i)\in [0,T)\times \R \times (0,+\infty) \times \mathcal E,\label{eq:hjb_back}\end{align}
where $\mathcal L^H$ is the infinitesimal generator given in \eqref{mgenxy}, with the terminal condition $V(T,x,s,e_i)=-e^{-\gamma x}$. We guess that the value function has the form $V(t,x,s,e_i)=-e^{-\gamma x +J_1(t)s^{-2\beta} +J_2(t,e_i)}$. Plugging this expression into \eqref{eq:hjb_back} and taking the first order condition on $\Pi$ yields \eqref{eq:strategia_backward_excomp}. The second order conditions guarantee that $\Pi^{B,*}$ is the optimal investment strategy. For the optimal reinsurance strategy  $\theta^{B,*}(t,e_i)$ we argue as in the proof of Proposition \ref{prop:optimal} and hence we get that  $\theta^{B,*}(t,e_i)=\theta^*(t,e_i)$  given in equation \eqref{optr}.\\
Next, we establish a verification result. Let $v(t,x,s,e_i)$ be a solution of the equation \eqref{eq:hjb_back} with the final condition $v(T,x,s,e_i)=-e^{-\gamma x}$ (that is $v(T,x,s,e_i)=V(T,x,s,e_i)$). Then, by It\^o's formula it holds that (for simplicity, we omit the dependence of $X$ on the strategy $H$)
\begin{align}
	&v(T, X_T, S_T,Y_T)=v(t,x,s,e_i)+\int_t^T\!\mathcal{L}^Hv(r, X_r,S_r,Y_r) \ud r \\ & + \int_t^T \!\Pi_r \sigma(Y_r)S_r^\beta\frac{\partial v}{\partial x}(r, X_r,S_r,Y_r) \ud W_r + \int_t^T\! \sigma(Y_r)S_r^{\beta+1}\frac{\partial v}{\partial s}(r, X_r,S_r,Y_r)  \ud W_r \\ & +  \int_t^T \! \sum_{j=1}^{K} \Big{\{} v(r, X_{r}, S_r,e_j)-v(r, X_{r}, S_r,Y_{r-}) \Big{\}} \left(m^Y-\nu^Y)(\ud r,\{e_j\} \right) \\&+ \int_t^T \!\!\!\!\int_{{\mathcal Z} }\!\! \big{\{} v(r, X_{r-\!}\!-\!(1-\theta_{r-\!})z, S_r,Y_r)-v(r, X_{r-\!}, S_r,Y_r) \big{\}} \! (m(\ud r, \ud z)-\lambda(r, Y_{r-}) F(\ud z) \ud r).
\end{align}
Since $v$ satisfies equation \eqref{eq:hjb_back}, we get that
\begin{align}
	&v(T, X_T, S_T,Y_T)\leq v(t,x,s,e_i) \\ &  + \int_t^T  \Pi_r \sigma(Y_r) S_r^\beta \frac{\partial v}{\partial x}(r, X_r,S_r, Y_r) \ud W_r + \int_t^T  \sigma(Y_r) S_r^{\beta+1} \frac{\partial v}{\partial s}(r, X_r,S_r, Y_r) \ud W_r \\
	&  +  \int_t^T \! \sum_{j=1}^{K} \Big{\{} v(r, X_{r}, S_r,e_j)-v(r, X_{r}, S_r,Y_{r-}) \Big{\}} \left(m^Y-\nu^Y)(\ud r,\{e_j\} \right) \\& + \int_t^T \!\!\!\!\int_{{\mathcal Z} }\!\! \big{\{} v(r, X_{r-\!}\!-\!(1-\theta_{r-\!})z, S_r,Y_r)-v(r, X_{r-\!}, S_r,Y_r) \big{\}} \! (m(\ud r, \ud z)-\lambda(r, Y_{r-}) F(\ud z) \ud r).\quad{}\ \label{eq:ineq:1}
\end{align}
Let
\begin{align}
	&M_t=\int_{t_0}^t  \Pi_r \sigma(Y_r)S_r^\beta \frac{\partial V}{\partial x}(r, X_r, S_r,Y_r) \ud W_r + \int_{t_0}^t \sigma(Y_r)S_r^{\beta+1} \frac{\partial V}{\partial s}(r, X_r, S_r,Y_r) \ud W_r \\
	&  +  \int_{t_0}^t \! \sum_{j=1}^{K} \Big{\{} V(r, X_{r}, S_r,e_j)-V(r, X_{r}, S_r,Y_{r-}) \Big{\}} \left(m^Y-\nu^Y)(\ud r,\{e_j\} \right) \\
	& + \int_{t_0}^t \!\!\int_{{\mathcal Z} }\!\! V(r, X_{r-}-(1-\theta_{r-})z, S_r,Y_r)-V(r, X_{r-}, S_r,Y_r) \left(m(\ud r, \ud z)-\lambda(r, Y_{r-}) F(\ud z) \ud r\right).
\end{align}
{If $M$ is an $(\bF,\P)$-martingale, then taking the conditional expectation given $X_t=x$, $S_t=s$, $Y_t=e_i$} on both sides of inequality \eqref{eq:ineq:1} yields
\[
V(t,x,s,e_i)\leq v(t,x,s,e_i),
\]
and the equality holds if $H$ is a maximizer of equation \eqref{eq:hjb_back}.
Then, it only remains to prove that the function $V(t,x,s,e_i)=-e^{-\gamma x +J_1(t)s^{-2\beta} +J_2(t,e_i)}$ is such that the process $M$ is an $(\bF,\P)$-martingale. To this aim, observe that $J_1(t)$ and $J_2(t,e_i)$ are both bounded in $[t_0,T]$ and we consider the localizing sequence of random times $\{\widetilde \tau_{n}\}_{n \in \bN}$ defined as
\[
\widetilde \tau_{n}:= \inf \Big{\{} t \ge t_0 : \ S_t^{-2\beta} > n ,\ X_t<-n  \Big{\}}, \quad n\in \bN.
\]
Then, $\{\widetilde \tau_{n}\}_{n \in \bN}$  is an increasing sequence, $\lim_{n \to \infty}\widetilde \tau_n \wedge T =T$ and hence we get that
\begin{align}
	&\mathbb{E}\bigg[\int_{t_0}^{T\wedge\widetilde\tau_n} \gamma^2\sigma^2  \Pi_r^2 S^{2\beta}_r V^2(r, X_r,S_r, Y_r) \ud r\bigg]\\
&+ \mathbb{E}\bigg[\int_{t_0}^{T\wedge\widetilde\tau_n} 4\beta^2\sigma^2  J_1^2(r) S^{-2\beta}_r V^2(r, X_r,S_r, Y_r) \ud r\bigg]\\
&+ \mathbb{E}\bigg[\int_{t_0}^{T\wedge\widetilde\tau_n} \Big{|}V(r, X_r,S_r, Y_{r-})\Big{|} \sum_{j=1}^{K}\left(e^{J_2(r, e_j)-J_2(r, Y_{r-})}-1\right)\nu^Y(\ud r,\{e_j\}) \bigg]\\
&+ \mathbb{E}\bigg[\int_{t_0}^{T\wedge\widetilde\tau_n} \lambda(r, Y_{r-}) \Big{|}V(r, X_{r-},S_r, Y_{r}) \Big{|} \sup_{z \in \mathcal{Z}} \left(e^{\gamma(1-\theta_r)z}-1\right) \ud r\bigg]<\infty,
\end{align}
which concludes the proof.



\end{document}